\documentclass{article}
\usepackage[utf8]{inputenc}
\usepackage[T1]{fontenc}
\usepackage{amsmath}
\usepackage{dsfont}
\usepackage{fancyhdr}
\usepackage[table,xcdraw]{xcolor}
\usepackage{graphicx}
\usepackage{multirow}
\usepackage{endnotes}
\usepackage{float}
\usepackage[nottoc]{tocbibind}
\usepackage{vmargin}
\usepackage{amssymb}
\usepackage{subfigure}
\usepackage{import}
\usepackage{wrapfig}
\usepackage{mathtools}
\usepackage{amsthm}
\usepackage{rotating}
\usepackage[title]{appendix}
\usepackage{listings}
\usepackage{soul}
\usepackage{stmaryrd}

\SetSymbolFont{stmry}{bold}{U}{stmry}{m}{n}

\usepackage{tikz}
\usepackage{bbm}
\usepackage{mleftright}
\usepackage{soul}
\mleftright
\usetikzlibrary{matrix}
\usetikzlibrary{shapes,arrows,chains, decorations.pathmorphing, decorations.pathreplacing,external}
\tikzset{quantum/.style={decorate, decoration=snake}}

\newcommand{\ket}[1]{\left|#1\right\rangle}

\newcommand{\ketbra}[2]{\left|#1\rangle\langle#2\right|}

\newcommand{\abs}[1]{\lvert #1\rvert}

\newcommand{\norm}[1]{\| #1\|}

\usepackage{cancel}

\newcommand{\QPVBBf}{$\mathrm{QPV}_{\mathrm{BB84}}^{f}$}

\newcommand{\tr}[1]{\mathrm{Tr}\left[#1\right]}
\newcommand{\ptr}[2]{\mathrm{Tr}_{#1}\left[#2\right]}

\usepackage{slashed}
\usepackage{enumitem}
\usepackage{authblk}

\theoremstyle{plain}
\newtheorem{theorem}{Theorem}[section]
\newtheorem*{theorem*}{Theorem}
\newtheorem{remark}[theorem]{Remark}
\newtheorem{prop}[theorem]{Proposition}
\newtheorem{definition}[theorem]{Definition}

\newtheorem{lemma}[theorem]{Lemma}
\newtheorem{corollary}[theorem]{Corollary}

\usepackage[normalem]{ulem}

\usetikzlibrary{decorations.pathreplacing,calligraphy}

\definecolor{darkred}{RGB}{179, 16, 32}
\definecolor{pine}{RGB}{1,121,111}

\usepackage[pdfencoding=auto, psdextra]{hyperref}
\hypersetup{
    colorlinks=true, 
    allcolors=darkred
}

\newcommand{\diagdots}[3][-25]{%
  \rotatebox{#1}{\makebox[0pt]{\makebox[#2]{\xleaders\hbox{$\cdot$\hskip#3}\hfill\kern0pt}}}%
}

\begin{document}

\title{Making Existing Quantum Position Verification Protocols Secure Against Arbitrary Transmission Loss
}

\author[1,2]{Rene Allerstorfer}
\author[3]{Andreas Bluhm}
\author[2,4,5]{Harry Buhrman}
\author[6]{Matthias Christandl}
\author[2,5]{Lloren\c{c} Escol\`a-Farr\`as}
\author[2,5]{Florian Speelman}
\author[1,2,7]{Philip Verduyn Lunel}

\affil[1]{CWI Amsterdam, CWI, Amsterdam, The Netherlands}
\affil[2]{QuSoft, CWI Amsterdam, The Netherlands}

\affil[3]{Univ. Grenoble Alpes, CNRS, Grenoble INP, LIG, France}
\affil[4]{Quantinuum, Partnership House, Carlisle Place, London, United Kingdom}
\affil[5]{University of Amsterdam, Amsterdam, The Netherlands}
\affil[6]{University of Copenhagen, Copenhagen, Denmark}
\affil[7]{Sorbonne Universit\'e, CNRS, LIP6, Paris, France}

\renewcommand\Affilfont{\itshape\small}
\maketitle
\begin{abstract}
\noindent Signal loss poses a significant threat to the security of quantum cryptography when the chosen protocol lacks loss-tolerance. In quantum position verification (QPV) protocols, even relatively small loss rates can compromise security. The goal is thus to find protocols that remain secure under practically achievable loss rates. In this work, we modify the usual structure of QPV protocols and prove that this modification makes the potentially high transmission loss between the verifiers and the prover security-irrelevant for a class of protocols that includes a practically-interesting candidate protocol inspired by the BB84 protocol (\QPVBBf). This modification, which involves photon presence detection, a small time delay at the prover, and a commitment to play before proceeding, reduces the overall loss rate to just the prover's laboratory. The adapted protocol \textsf{c-}\QPVBBf~then becomes a practically feasible QPV protocol with strong security guarantees, even against attackers using adaptive strategies. As the loss rate between the verifiers and prover is mainly dictated by the distance between them, secure QPV over longer distances becomes possible. We also show possible implementations of the required photon presence detection, making \textsf{c-}\QPVBBf~a protocol that solves all major practical issues in QPV. Finally, we discuss experimental aspects and give parameter estimations.
\end{abstract}

\tableofcontents

\section{Introduction}

Imagine the following situation: You are sitting in front of your computer screen, looking at a website that looks like the website of your bank. But how can you make sure it is authentic? One way would be to verify that the server is indeed placed in the basement of your bank. Or, say, you see a photo or video online. One way to verify that it is most likely authentic, rather than fake or created by powerful AI, would be to verify \textit{where} it was created. That is the idea behind position-based cryptography, in which the geographic location of a party is used to authenticate it, without further cryptographic assumptions. 

The fundamental building block for this is, as in the example above, secure \emph{position verification}. For simplicity, we will focus in this article on the one-dimensional case, in which two verifiers ($V_0$ and $V_1$) want to securely verify the position $z$ of a prover ($P$) located between them. In particular, they need to be able to distinguish the honest situation from the case in which no-one is at the location to be verified, but two attackers (Alice and Bob) try to fool the verifiers while Alice is placed between $V_0$ and $z$ and Bob between $z$ and $V_1$.

Unfortunately, secure position verification with classical resources is impossible without further assumptions as shown in \cite{chandran_position_2009}, since classical information can be copied and therefore easily be distributed among the attackers. Quantum information, however, cannot be copied perfectly. This motivated the study of quantum information protocols for secure position verification, or quantum position verification (QPV) for short. The first proposals to this end resulted in a patent by Beausoleil, Kent, Munro, and Spiller published in 2006 \cite{KentPatent2006}. More proposals that were claimed to be secure followed in the academic literature in 2010~\cite{Malaney2010, MalaneyNoisy2010}. However, first ad-hoc attacks were found to compromise the security of these protocols~\cite{kent_quantum_2011,lau_insecurity_2011}, before a general attack on any QPV protocol was put forward by Buhrman, Chandran, Fehr, Gelles, Goyal, Ostrovsky, and Schaffner~\cite{buhrman_position-based_2011}. The attack makes ingenious use of quantum teleportation and requires a doubly exponential amount of pre-shared entangled pairs. This amount was later reduced to exponential by Beigi and K\"onig~\cite{beigi_simplified_2011} with the help of port-based teleportation \cite{grinko2023gelfand, fei2023efficient, grinko2023efficient}. This idea was subsequently generalized to other settings in~\cite{gao2013,gao2016quantum,dolev2019}.

While these results have proved that unconditionally secure protocols for QPV are impossible, the aim shifted to proving practical security of QPV protocols. Since it is hard to generate and maintain entanglement, it would be enough to find protocols which need an unrealistically large amount of entanglement to attack them to have information-theoretic security in practice. Therefore, the main interest at present is to consider security against bounded attackers. For example, the QPV$_{\text{BB84}}$ protocol \cite{kent_quantum_2011}, inspired by the BB84 quantum key-distribution protocol, involves only a single qubit sent by $V_0$ in one of the four BB84 states $\ket{0}$, $\ket{1}$, $\ket{+}$, or $\ket{-}$. This protocol is secure against unentangled attackers \cite{buhrman_position-based_2011}, but can be broken by attackers sharing a single entangled pair~\cite{lau_insecurity_2011}. However, this protocol allows for parallel repetition, such that $\Theta({n})$ entangled pairs are required to break its $n$-fold parallel repetition \cite{tomamichel_monogamy--entanglement_2013,ribeiro_tight_2015}. In practice, the fact that the entanglement needed scales with the amount of rounds played in parallel is not a very strong security guarantee, since the honest prover also needs to manipulate an equal amount of qubits as there are rounds. Ideally, we would like to find protocols where the honest prover has to manipulate a small quantum system, while the attackers need to pre-share a very large entangled state, i.e., many EPR pairs. Significant progress to this problem was made in \cite{bluhm2022single}, with a different version of the protocol, QPV$^f_{\text{BB84}}$. Here the basis in which the honest prover needs to apply his measurement is determined by a classical function $f$ depending on two $n$-bit input strings $x,y$. In the paper the authors prove security against $\Omega({n})$ entangled pairs pre-shared by the attackers for a random function $f$. Note that in this protocol there is only a single qubit, but the required quantum resources for an attack scale at least linearly in the classical information sent. For an honest prover it is much easier to do some computation on classical inputs than on quantum inputs. It has the additional advantage of being secure even with slowly traveling qubits, as for example qubits sent over optical fiber, where transmission speed is typically $2/3$ the speed of light. Moreover, in a future quantum network it will likely often be the case that there is no direct link between the verifiers and the prover wanting to run a QPV protocol, further emphasizing the need for protocols that can deal with slow quantum information. Other protocols combining classical and quantum information can be found in~\cite{kent_quantum_2011,chakraborty_practical_2015, unruh_quantum_2014,junge2021geometry,qi_loss-tolerant_2015,allerstorfer2023security}. Attacks for such protocols have also been analyzed in~\cite{buhrman2013garden, speelman2016,olivo_breaking_2020}. In particular, in a recent breakthrough by some of the present authors, subexponential upper bounds have been proved for attacks on the qubit routing protocol based on conditional disclosure of secrets schemes \cite{allerstorfer2023relating}.  Alternative models of security use oracles \cite{unruh_quantum_2014} or computational assumptions \cite{liu2022beating}.

Although the protocol QPV$^f_{\text{BB84}}$ is also resistant against small amounts of noise and loss as shown in \cite{bluhm2022single,escolafarras2022singlequbit}, none of the above protocols is proved secure under conditions consistent with current technologies, where the main source of error is photon loss. Using optical fibre, photon transmission decays exponentially in the distance and at some point almost all photons will be lost. This can compromise security in QPV protocols that are not loss tolerant, and immediately makes QPV$^f_{\text{BB84}}$ insecure in basically any practical setting. This is a major downside of QPV$^f_{\text{BB84}}$, since apart from this issue it has the most desirable properties of all known proposed protocols.

A common approach to deal with photon loss is to disregard rounds in which the prover claims that a photon was lost during transmission. Regrettably, this approach renders these protocols vulnerable to attackers since the attacks can take advantage of the photon loss by claiming the photon was lost if they risk being detected. Recent progress towards addressing this major obstacle to protocols that can be implemented on current devices has been made in \cite{allerstorfer2021towards, allerstorfer2022role}, where fully loss-tolerant protocols were studied. However, those protocols were found to be vulnerable against simple entanglement-based attacks. And even though loss is not an issue in \cite{liu2022beating} as all the communication is classical, their protocol requires a large quantum computer at the prover to prepare the states used in it and therefore is not viable in the near-term. So far, however, a protocol has been lacking that is both provably secure against realistic attacks while still being implementable with current technologies.

\subsection{Results}

In our contribution, we focus on the design of such a practically feasible and secure QPV protocol. We introduce a structural modification to QPV where, instead of the verifiers sending the information to the prover such that all information arrives at the same time, the quantum information shall arrive slightly before the classical information. The prover confirms the reception of the quantum information, and \emph{commits} to playing, after which he receives the classical information to complete the task. In this way, for a QPV protocol $\mathsf P$, we define its \emph{committing} version \textsf{c-}$\mathsf{P}$.

Consider a secure QPV protocol $\mathsf{P}$ with classical prover responses, which remains secure when played in sequential repetition and in which the honest quantum information is allowed to travel slowly (like \QPVBBf). This implies that the protocol is \emph{state-independent}, in the sense that the attackers can replace the input state with any other quantum state. Then our main result states that for every such QPV protocol $\mathsf{P}$, its committing version \textsf{c-}$\mathsf{P}$ inherits the security of $\mathsf{P}$, while becoming fully loss-tolerant against transmission loss. Denoting by $\eta_V$ the transmission rate from the verifiers to the prover and by $\eta_P$ the one within the prover's laboratory (between committing and receiving the classical information), we informally state our main result, Theorem~\ref{thm: main theorem upper bounding commit attack}, as follows:

\begin{theorem*}[Informal]\label{thm:theorem_committing_intro}
The success probability of successfully attacking \emph{\textsf{c-}}$\mathsf{P}$ (with both $\eta_V$ and $\eta_P$) reduces to the probability of attacking $\mathsf{P}$ (with only $\eta_P)$:
\begin{align}
        \mathbb P[\mathrm{attack } \, \emph{\textsf{c-}}\mathsf {P}_{\eta_{V},\eta_P}] \leq \mathbb P[\mathrm{attack } \, \mathsf P_{\eta_{P}}]+\frac{1}{k},
    \end{align}
where $k$ is a parameter defined by the number of rounds $N$ in which both attackers commit and grows roughly as $O(N^{\frac{1}{7}})$.
\end{theorem*}

This means that the potentially very high loss between the verifiers and the prover, $1-\eta_V$, becomes irrelevant to security in \textsf{c-}$\mathsf{P}_{\eta_{V},\eta_P}$ and only the much smaller loss at the prover's laboratory, $1-\eta_P$, matters. And for sufficiently high values of $\eta_P$ we often have security guarantees, e.g.\ for \QPVBBf~\cite{bluhm2022single, escolafarras2022singlequbit}. In theory, for an ideal prover, \textsf{c-}$\mathsf{P}_{\eta_{V},\eta_P}$ becomes fully loss-tolerant.

If we demand perfect coordination in commitments for all possible inputs, which is expected from the honest prover, then our result reduces to
\begin{align}\label{thm:pattackequal_intro}
    \mathbb P[\mathrm{attack } \, \textsf{c-}\mathsf{P}_{\eta_{V},\eta_P}] = \mathbb P[\mathrm{attack } \, \mathsf P_{\eta_{P}}],
\end{align}
as the other direction $\mathbb P[\mathrm{attack } \, \mathsf P_{\eta_{P}}] \leq \mathbb P[\mathrm{attack } \, \textsf{c-}\mathsf {P}_{\eta_{V},\eta_P}]$ is simple to see\footnote{The attackers can just pre-agree to commit with a rate $\eta_V$ and use the strategy of $\mathsf P_{\eta_{P}}$ to produce the answers for \textsf{c-}$\mathsf{P}_{\eta_{V},\eta_P}$.}. The above theorem allows us to make our argument robust, as very small values of incorrect coordination (relative to the $2^{2n}$ input pairs $x,y$) could, in principle, help attackers, while leaving them undetected with high probability.

We further prove that the success probability for attacking our protocol decays exponentially with the number of (sequentially repeated) rounds run, even if attackers are allowed to use adaptive strategies.

Applying our results to \QPVBBf, we show that quantum position verification is possible even if the loss is arbitrarily high, the (constant-sized) quantum information is arbitrarily slow, and attackers pre-share some entanglement (bounded in the classical message length $n$). The question of a super-linear lower bound on the required resources for a quantum attack still remains open.

Finally, we study two possible ways of implementing the non-demolition photon presence detection step of our protocol: true photon presence detection as demonstrated in \cite{niemietz2021nondestructive} as a potential long-term solution, and a simplified photon presence detection based on a partial Bell measurement \cite{michler1996interferometric} at the prover that is technologically feasible today. In the latter, the honest prover essentially teleports the input state of the protocol to himself and concludes the presence of that state based on a conclusive click pattern in the partial Bell measurement, in which case the quantum state got teleported and can be further acted on by the prover (e.g.~by a polarization measurement). We note that for the committing version of \QPVBBf,  \textsf{c-}\QPVBBf, no active feed-forward for the teleportation corrections is required, as they predictably alter the subsequent measurement outcome and thus can be classically corrected by the prover post-measurement. We identify the experimental requirements at the prover as: being able to generate an EPR pair, to do a partial Bell measurement, to store the teleported quantum state in a short delay loop until the classical input information $(x,y)$ arrives, and the ability to perform the protocol measurement based on $(x,y)$. The latter shall be possible fast enough such that the protocol rounds can be run with high frequency (say, MHz or ideally GHz). To that end, we argue that with top equipment, MHz rates are currently possible and GHz rates feasible in principle. We estimate from our results that approximately $10^9$ rounds are required to decisively verify a prover's location. Consequently, at a clock rate of 100 MHz, it would take around 10 seconds to run the protocol. Such a duration could still be sufficient for position verification, as a physical object cannot move significantly in a few seconds. We further note that our bounds could potentially be improved to require fewer rounds. Practically, also the signal-to-noise ratio of the photon presence detection is an important figure of merit that is relevant for the security of the protocol, which we discuss further in the experimental section of the paper. We argue that with state-of-the-art equipment our protocol can remain within its secure regime, even in practice.\footnote{As the numbers will strongly depend on the actual experimental setup of a demonstration, we only give estimations.}

To summarize, our main result holds more generally, but applied to \QPVBBf~we provide a new QPV protocol, \textsf{c-}\QPVBBf, that is a practically feasible QPV protocol with decent security guarantees in the most general setting, even in practice. This opens up the road for a first experimental demonstration of quantum position verification.

\section{Preliminaries}
Let $\mathcal{H}$, $\mathcal{H'}$ be finite-dimensional Hilbert spaces. We denote by $\mathcal{B}(\mathcal{H},\mathcal{H'})$ the set of bounded operators from $\mathcal{H}$ to $\mathcal{H'}$ and $\mathcal{B}(\mathcal{H})=\mathcal{B}(\mathcal{H},\mathcal{H})$.  Denote by $\mathcal{S}(\mathcal{H})$ the set of quantum states on $\mathcal{H}$,~i.e.\ $\mathcal{S}(\mathcal{H})=\{\rho\in\mathcal{B}(\mathcal{H})\mid \rho\geq0, \tr{\rho}=1)\}$. For $\rho,\sigma\in\mathcal{B}(\mathcal{H})$, a measure of distance between them is
\begin{equation}
    || \rho-\sigma||_1:=\tr{\sqrt{(\rho-\sigma)(\rho-\sigma)^{\dagger}}}.
\end{equation}
A linear map $\mathcal{E}:\mathcal{B}(\mathcal H)\rightarrow \mathcal{B}(\mathcal H')$ is a \emph{quantum channel} if it is completely positive and trace preserving (CPTP). 

\begin{lemma} \label{lemma thm Kraus decomposition} \emph{(Kraus representation \cite{Kraus1971311})}. A linear map $\Phi$ is completely positive and trace non-increasing if and only if there exist bounded operators $\{K_i\}_{i=1}^r$ such that for all density operators $\rho$,
\begin{equation}
    \Phi(\rho)=\sum_{i=1}^rK_i\rho K_i^{\dagger},
\end{equation}
with $\sum_{i=1}^rK_i^{\dagger} K_i\leq \mathbb{I}$, where $r$ is the Kraus rank. Moreover, $\Phi$ is trace-preserving,~i.e.\ a quantum channel, if and only if $\sum_{i=1}^rK_i^{\dagger} K_i = \mathbbm{1}$.
\end{lemma}
Let $\Omega$ be a finite outcome set. A \emph{quantum instrument} $\mathcal{I}$ is a set of completely positive linear maps $\{\mathcal{I}_i\}_{i\in\Omega}$ such that $\sum_{i\in\Omega}\mathcal{I}_i$ is trace preserving. Given the quantum state $\mathcal{\rho}\in\mathcal{S}(\mathcal{H})$, the probability of obtaining outcome $i$ is given by $\tr{\mathcal{I}_i(\rho)}$ and the sub-normalized output state upon outcome $i$ is $\mathcal{I}_i(\rho)$.

\subsection{Introduction to QPV}

All proposed QPV protocols rely on both relativistic constraints and the laws of quantum mechanics for their security. The QPV literature usually focuses on the 1-dimensional case, so verifying the position of a prover $P$ on a line, as it makes the analysis easier and the main ideas generalize to higher dimensions.

The usual general setting for a 1-dimensional QPV protocol is the following: two verifiers $V_0$ and $V_1$, placed on the left and right of $P$, send quantum and/or classical messages to $P$ at the speed of light. $P$ has to pass a challenge and to reply correctly to them with a signal at the speed of light as well. The verifiers have perfectly synchronized clocks and if any of them receives an inconsistent answer or if the timing of the answers is not as expected from the honest prover, they abort the protocol\footnote{The time consumed by the prover to perform the task is assumed to be negligible relative to the total protocol time}.

We will mainly focus on one type of QPV protocol, \QPVBBf~\cite{bluhm2022single}. This protocol is well studied, easy to implement and the lower bounds on the required quantum resources to attack them scale linearly in the classical input size. However, it is not loss-tolerant enough for practical purposes. We set out to solve this issue in this work.

\begin{remark}
    We describe the \QPVBBf~protocol in its \emph{purified} version, where a verifier sends half of an EPR pair instead of a single qubit, as they would do in its \emph{prepare-and-measure} version. Both versions are equivalent, but we use the purified version for our proof analysis. 
\end{remark}

\begin{definition} \label{def qpv bb84 f} \emph{(\QPVBBf~protocol \cite{bluhm2022single,escolafarras2022singlequbit})}.
Let $n\in\mathbb{N}$, and consider a $2n$-bit boolean function $f:\{0,1\}^n \times \{0,1\}^n \to \{0,1\}$. A round of the \QPVBBf~protocol is described as follows.
\begin{enumerate}
    \item $V_0$ prepares the EPR pair $\ket{\Phi^+}=(\ket{00}+\ket{11})/\sqrt{2}$ and sends one qubit $Q$ of $\ket{\Phi^+}$ and $x\in\{0,1\}^n$ to $P$ and $V_1$ sends $y\in\{0,1\}^n$ to $P$ such that all information arrives at $P$ simultaneously. The classical information is required to travel at the speed of light, the quantum information can be sent arbitrarily slowly.
    \item Immediately, $P$ measures $Q$ in the basis $f(x,y)$\footnote{Usually, the two bases correspond to the computational and the Hadamard basis, justifying the nomenclature of \QPVBBf. If $m$ basis choices are possible, the range of $f$ will be $\{0, 1, \dots, m-1\}$.} and broadcasts his outcome $a \in \{0, 1 \}$ to $V_0$ and $V_1$. If the photon is lost, he sends `$\perp$'.
    \item The verifiers measure the qubit they kept in the basis $f(x,y)$, getting outcome $v \in \{ 0,1 \}$. They accept if $a=v$ and $a$ arrives on time. They record `photon loss' if they both receive `$\perp$' on time. If either the answers do not arrive on time or are different, the verifiers abort.  
\end{enumerate}
\end{definition}

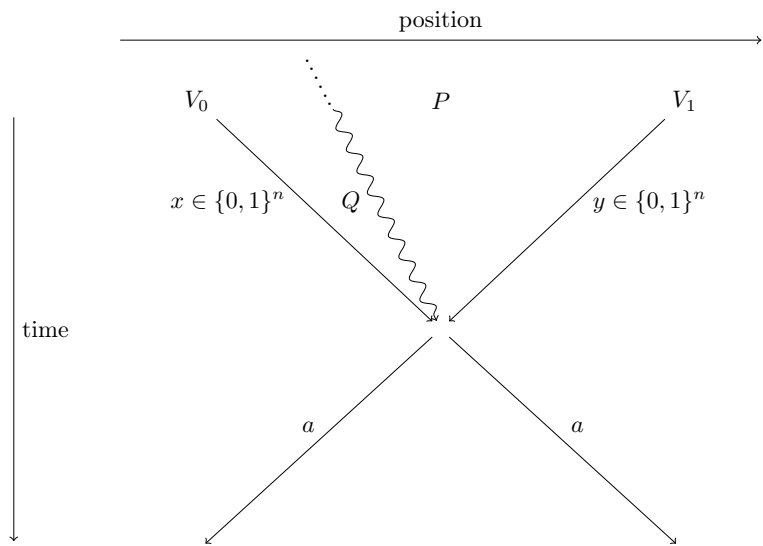
\begin{figure}[htbp]
    \centering
    \scalebox{0.9}{
    \begin{tikzpicture}[node distance=3cm, auto]
    \node (A) {$V_0$};
    \node [left=1cm of A] {};
    \node [right=of A] (B) {$P$};
    \node [right=of B] (C) {$V_1$};
    \node [right=1cm of C] {};
    \node [below=of A] (D) {};
    \node [below=of B] (E) {};

    \node [above=-0.1cm of A] (N) {};
    \node [right=1.46cm of N] {$\diagdots[117]{2.5em}{0.1em}$};
    
    \node [right=1.5cm of A] (M) {};
    \node [left=1.5cm of C] (M2) {};
    
    \node [below=of C] (F) {};
    \node [below=of D] (G) {};
    \node [below=of E] (H) {};
    \node [below=of F] (I) {};
    \node [left= 6cm of E] (J) {};
    \node [below= 3cm of J] (K) {};
    \node [above= 3cm of J] (L) {};

    \draw [->, transform canvas={xshift=0pt, yshift=0pt}, quantum] (M) -- (E) node[midway] (x) {} ;
    \draw [->] (A) -- (E);
    \draw [->] (C) -- (E);
    \draw [][->] (E) -- (I) node[midway] (q) {$a$}; 
    \draw [][->] (E) -- (G); 

    \draw [->] (L) -- (K) node[midway] {time};

    \node[left=0.3cm of x, transform canvas={xshift=+ 2pt, yshift = +2 pt}] {$Q$};
    \node[left=1.4cm of x, transform canvas={xshift=+ 2pt, yshift = +2 pt}] {$x\in\{0,1\}^n$};
    \node[right = 2.6cm of x, transform canvas={xshift=+ 2pt, yshift = +2 pt}] {$y \in \{0,1\}^n$};
    \node[left = 3.5cm of q] {$a$};

    \node [above=0.5cm of A] (posV00) {};
    \node [left=1cm of posV00] (posV0) {};
    \node [above=0.5cm of C] (posV11) {};
    \node [right=1cm of posV11] (posV1) {};
    \draw [->] (posV0) -- (posV1) node[midway] {position};

    \node [right=0.8cm of C] (VP0) {};
    \node [right=4.55cm of E] (VPP) {};
    \node [right=0.8cm of I] (PV) {};
    \end{tikzpicture}
    }
\caption{Schematic representation of the \QPVBBf~protocol. Undulated lines represent quantum information, whereas straight lines represent classical information. The slowly travelling quantum system $Q$ originated from $V_0$ in the past.}
\label{fig:general classical and quantum info protocol}
\end{figure}

In the end, the verifiers accept the location of the prover $P$ if after multiple repetitions of single rounds they receive answers that are consistent with their known experimental parameters, i.e.\ if the number of `photon loss' answers is consistent with the transmission rate $\eta$,  and the number of wrong answers is consistent with the error in the experimental set-up. 

\subsubsection*{General structure of an attack on \QPVBBf}
In a general attack on the \QPVBBf~protocol, Alice and Bob act as follows.

\begin{enumerate}
    \item The attackers prepare a joint (possibly entangled) quantum state. 
    \item Alice intercepts the quantum information sent from $V_0$ and performs an arbitrary quantum channel. She keeps a part of the resulting state and sends the rest to Bob. Denote by $\rho$ their joint state at this stage (before communication).  
    \item Alice and Bob intercept $x$ and $y$, make a copy and send it to the other attacker, respectively. Both then can apply local quantum channels depending on $x$ (at Alice) and $y$ (at Bob) to $\rho$. Each can keep part of the resulting local state and send the other part to their fellow attacker.
    \item Upon receiving the information sent by the other party, each attacker can locally apply an arbitrary POVM depending on $(x,y)$ to obtain classical answers, which will be sent to $V_0$ and $V_1$, respectively.
\end{enumerate}
If there is loss in the protocol the attackers need to mimic the transmission rate of the prover.

\subsubsection*{Known properties of \QPVBBf}

Neglecting photon loss, \QPVBBf~was proven to be secure \cite{bluhm2022single} even if attackers pre-share a linear amount of qubits in the size of the classical information $n$. The main advantage of this protocol is that it only requires sending a single qubit whereas adversaries using an increasing amount of entanglement can be combatted solely by increasing the number of classical bits used in the protocol. In addition, \QPVBBf~has the advantage that the quantum information can travel arbitrarily slowly. However, photon loss constitutes a major problem. Consider the following easy-to-perform attack, where Alice makes a random guess for the value of $f(x,y)$ and just measures in the guessed basis and broadcasts the result to Bob. Both attackers intercept the classical information, make a copy and send it to their fellow attacker. After one round of simultaneous communication, each can compute $f(x,y)$ and both know if the initial guess was correct. If so, they send the outcome of the measurement, which is correct, to the verifiers. Otherwise, they claim no photon arrived. Alice's basis guess will be correct half of the time (or $1/m$ of the time for more basis choices) and therefore, if the transmission rate is such that $\eta\leq\frac{1}{2}$ (or $1/m$, respectively), the attackers will be correct whenever they answer and thus break the protocol. 

In \cite{escolafarras2022singlequbit}, the range $1/2 < \eta \leq 1$ was studied for \QPVBBf, and it was shown that the protocol remains secure for attackers who pre-share a linear amount of entanglement in $n$ and arbitrary slow quantum information. However, $\eta>\frac{1}{2}$ is only attainable for short distances. A way to bypass this, first shown independently in \cite{qi_loss-tolerant_2015} and \cite[Chapter 5]{FlorianThesis}, can be achieved by encoding the qubit $Q$ in more bases than just the computational and the Hadamard bases. In the first case, $Q$ is encoded in a uniformly random basis in the Bloch sphere, and security holds for reasonably high loss if the quantum information is sent at the speed of light and the attackers do not pre-share entanglement. Following the second approach, where $Q$ is encoded in $m$ bases in the Bloch sphere, \cite{escolafarras2022singlequbit} showed via semidefinite programming (whose size depends on $m$) that one can improve the loss-tolerance by increasing $m$, while preserving security against attackers who pre-share a linear amount of entanglement in $n$ and arbitrary slow quantum information. The specific cases of $m=3,5$ were worked out, showing that the protocol remains secure, preserving the other two properties, if up to 70\% of the photons are lost, making slightly larger distances than with two bases still feasible. 

In the next sections, we show how to make QPV for longer distances possible by slightly modifying the structure of the previously known protocols. This opens up a feasible route to the first experimental demonstration of a QPV protocol that captures security against the three major problems that the field faces: bounded attackers, photon loss (for large distances) and slow quantum information.

\section{QPV with a commitment}
One of the major issues in practical quantum cryptography is the transmission loss between the interacting parties. In the context of QPV a high loss between the verifiers and the prover can compromise security if the QPV protocol is not loss tolerant. Most QPV protocols are not loss tolerant, and the ones who are have other drawbacks, most notably being broken by an entanglement attack using only one pre-shared EPR pair \cite{lim_loss-tolerant_2016, allerstorfer2021towards} or requiring a large quantum computer at the prover and computational assumptions \cite{liu2022beating}.

To overcome this, we introduce the following modification to the structure of a certain class of QPV protocols. Let $\mathsf P_{\eta_V, \eta_P}$ be a QPV protocol with the verifiers sending quantum and classical information and the prover sending classical answers, where $\eta_V$ is the transmission rate between the verifiers and the prover, and $\eta_P$ is the transmission rate in the prover's laboratory.  We define its \emph{committing} version (or protocol with \emph{commitment}), denoted by \textsf{c-}$\mathsf P_{\eta_V,\eta_P}$, by introducing a small time delay $\delta>0$ between the arrival time of the quantum information and the classical information at the prover. When the quantum information arrives at $P$, he is required to commit to play ($c=1$) or not to play ($c=0$) the round. Only the $c=1$ rounds are later analyzed for security purposes. We will show that introducing this step will eliminate the relevance of the transmission rate $\eta_V$ from the verifiers to the prover for security. We prove that only the (potentially small) loss in the prover's laboratory $\eta_P$ will count now because of this post-selection on ``committed'' rounds. 

This trick can be applied to a class of QPV protocol that fulfills the necessary criteria of our proof. For concreteness, and because it is practically most interesting, we will focus on the case $\mathsf P_{\eta_V, \eta_P}=$ \QPVBBf, where we denote by \textsf{c-}\QPVBBf~the protocol with commitment.

\subsection{The protocol \textsf{c-}QPV\texorpdfstring{$^f_{\mathrm{BB84}}$}{TEXT}}
\noindent The \emph{committing} version of \QPVBBf~is described as follows. Again, we describe the protocol in its purified form, whereas in practice it might be simpler to implement its prepare-and-measure version.

\begin{definition} \label{def c_qpv bb84 f} Let $n\in\mathbb{N}$, and consider a $2n$-bit boolean function $f:\{0,1\}^n \times \{0,1\}^n \to \{0,1\}$. A round of the \QPVBBf~protocol with commitment, denoted by \emph{\textsf{c-}}\QPVBBf, is described as follows.

\begin{enumerate}
    \item $V_0$ prepares the EPR pair $\ket{\Phi^+}=(\ket{00}+\ket{11})/\sqrt{2}$ and sends one qubit $Q$ and $x\in\{0,1\}^n$ to $P$ and $V_1$ sends $y\in\{0,1\}^n$ to $P$ such that $x, y$ arrive a time $\delta > 0$ after $Q$ at $P$. The classical information is required to travel at the speed of light, the quantum information can be sent arbitrarily slowly.
    \item If the prover receives $Q$, he immediately confirms that and broadcasts the commitment bit $c=1$. Otherwise, he broadcasts $c=0$.
    \item If $c=1$, $P$ measures $Q$ in the basis $f(x,y)$\footnote{Again, for more basis choices, the range of $f$ would become $\{0, 1, \dots, m-1\}$.} as soon as $x, y$ arrive and broadcasts his outcome $a$ to $V_0$ and $V_1$. If the photon is lost in the time $\delta$ or during the measurement, he sends `$\perp$'.
    \item The verifiers collect $(c,a)$ and $V_0$ measures the qubit he kept in basis $f(x,y)$, getting result $v$. If $c=0$ they ignore the round. If $c=1$ they check whether $a = v$. If $c, a$ arrived at their appropriate times and $a=v$, they accept. They record `photon loss' if they both receive `$\perp$' on time. If any of the answers do not arrive on time or are different the verifiers abort.  
\end{enumerate}
\end{definition}

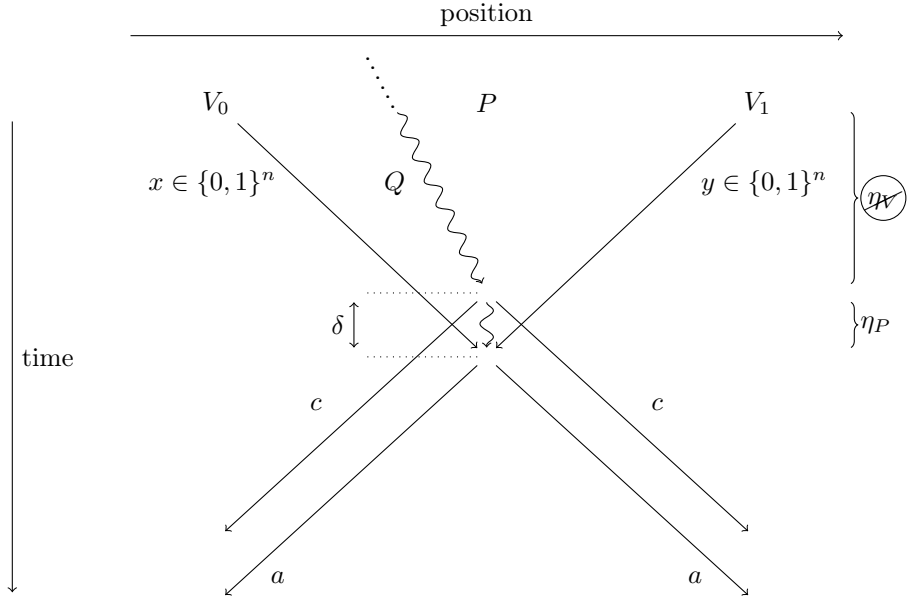
\begin{figure}[htbp]
    \centering
    \begin{tikzpicture}[node distance=3cm, auto]
    \node (A) {$V_0$};
    \node [left=1cm of A] {};
    \node [right=of A] (B) {$P$};
    \node [right=of B] (C) {$V_1$};
    \node [right=0.8cm of C] (VP0) {};
    
    \node [below=of A] (D) {};
    \node [below=of B] (E) {};
    \node [left=3cm of A] (T00) {};
    \node [below=0.01cm of T00] (T0) {};
    \node [below=0.01cm of C] (T0B) {};
    \node [above=0.5cm of A] (posV00) {};
    \node [left=1cm of posV00] (posV0) {};
    \node [above=0.5cm of C] (posV11) {};
    \node [right=1cm of posV11] (posV1) {};

    \node [above=-0.1cm of A] (N) {};
    \node [right=1.85cm of N] {$\diagdots[117]{2.5em}{0.1em}$};
    
    \node [right=1.9cm of A] (M) {};
    \node [left=1.9cm of C] (M2) {};
    
    \node [below=of C] (F) {};
    \node [below=of D] (G) {};
    \node [below=of E] (H) {};
    \node [below=of F] (I) {};
    \node [left= 6cm of E] (J) {};
    \node [below= 3cm of J] (K) {};
    \node [above= 3cm of J] (L) {};
    
    \node [above= 0.6cm of E] (P0) {};
    \node [right=4.55cm of P0] (VP1) {};
    \node [right=4.55cm of E] (VPP) {};
    
    \node [left=1.5cm of P0] (T12delta){};
    \node [left=1.5cm of E] (T12){};

    \draw [<->] (T12) -- (T12delta) node[midway] {$\delta$};
    
    \node [above= 0.6cm of I] (V0bis) {};
    \node [above= 0.6cm of G] (V1bis) {};
    
    \node [left= 2cm of E] (CommitA) {};
    \node [below= 0.3cm of CommitA] (CommitAbis) {$c$};
    \node [right= 2cm of E] (CommitB) {};
    \node [below= 0.3cm of CommitB] (CommitBbis) {$c$};
    
    \node [left= 2.5cm of E] (answerA) {};
    \node [below= 2.6cm of answerA] (answerAbis) {$a$};
    \node [left=3cm of G] (T1) {};
    \node [right= 2.5cm of E] (answerB) {};
    \node [below= 2.6cm of answerB] (answerABbis) {$a$};

    \draw[dotted][-](P0)--(T12delta){};
    \draw[dotted][-](E)--(T12){};

    \draw [decorate, decoration = {brace}] (VP0) -- (VP1) node[midway] { $\eta_{V}$};
    \draw [decorate, decoration = {brace}] (VP1) -- (VPP) node[midway] {$\eta_{P}$};

    \draw [->, transform canvas={xshift=0pt, yshift = 0 pt}, quantum] (M) -- (P0) node[midway] (x) {} ;

    \draw [->] (posV0) -- (posV1) node[midway] {position};
    \draw [->] (A) -- (E);
    \draw [->] (C) -- (E) node[midway] (question_y) {};
    \draw [][->] (E) -- (I) node[] (q) {}; 
    \draw [][->] (E) -- (G) node[] (q2) {}; 

    \node [right=1cm of I] (PV) {};
    
    \draw [->,transform canvas={xshift=0pt, yshift = 0 pt}, quantum]  (P0) -- (E) node[midway] (try) {};
    \draw [->] (P0) -- (V0bis) node[] (ca) {};
    \draw [->] (P0) -- (V1bis) node[] (cb) {};
    
    \draw [->] (L) -- (K) node[midway] {time};

    \node[left=0.4cm of x, transform canvas={xshift=+ 2pt, yshift = +2 pt}] {$Q$};
    \node[left=2.1cm of x, transform canvas={xshift=+ 2pt, yshift = +2 pt}] {$x\in\{0,1\}^n$};
    \node[right = 3cm of x, transform canvas={xshift=+ 2pt, yshift = +2 pt}] {$y \in \{0,1\}^n$};
\end{tikzpicture}
\caption{Schematic representation of the \textsf{c-}\QPVBBf~protocol. Undulated lines represent quantum information, straight lines represent classical information. The slowly traveling quantum system $Q$ originated from $V_0$ in the past. The novel aspects are the time delay $\delta > 0$ at the prover and the prover commitment $c \in \{ 0, 1 \}$. We show that for the security of this protocol, the transmission $\eta_V$ becomes irrelevant.}
\label{fig:protocol-cBB84f}
\end{figure}

\section{Security of QPV with commitment}
We consider a quantum position verification protocol $\mathsf P$ in a 1-dimensional setup such that (i) the verifiers send classical and quantum inputs to the prover, and (ii) the answers are classical, and the verifiers (iii) expect to receive `correct' ($\textsc{c})$, `incorrect' ($\textsc{i}$), and `no-photon' ($\perp$). If we want to make explicit that the transmission rate of the quantum information sent by the verifiers is $\eta$, we will denote the protocol by $\mathsf P_\eta$. The most general attack on $\mathsf P$  is to place an adversary, who we will call Alice, between $V_0$ and the position where the prover should be and another adversary, who we will call Bob, between the supposed prover location and $V_1$. It is easy to see that having more than two adversaries in a 1-dimensional setting does not improve an attack.  Before the protocol, the attackers prepare a joint (entangled) quantum state $\sigma$. Then, Alice and Bob intercept the information sent from the verifier closest to them, they make a copy and broadcast the classical information to their fellow attacker. Subsequently, they perform a quantum operation on the intercepted quantum information, keep a register and send another register to the other attacker. After one round of simultaneous communication, they both perform a POVM to obtain a classical answer, and they send it to their closest verifier, respectively. 

Denote by $x$ and $y$ the classical information sent from $V_0$ and $V_1$, respectively. Without loss of generality, consider them to be $n$-bit strings, and assume they are uniformly distributed. Denote by $\omega^{(x,y)}$ the quantum state after communication to which the attackers apply the POVM that gives the final answer. Fix a partition into systems $AA_\text{com}BB_\text{com}$, where `com' denotes the subsystems that will be communicated. We can write the attackers' POVMs as $\{ \Pi^{A,(x,y)}_{AB_\text{com},a}\}_{a\in\{0,1,\perp\}}$ and $\{ \Pi^{B,(x,y)}_{A_\text{com}B,b}\}_{b\in\{0,1,\perp\}}$, where we associate the outcomes to  `correct', `incorrect', and `no-photon' answers, denoted by $\textsc{c},\textsc{i},\perp$, respectively. 
Then, the probability that the attackers give the correct answers {while mimicking the response rate $\eta$} (and thus successfully attack the protocol) can be written as

\begin{equation}
    \mathbb P\left[\mathrm{attack }\hspace{1mm} \mathsf P_{\eta}\right]=\frac{1}{\eta}\frac{1}{2^{2n} }\sum_{x,y}\tr{ \left( \Pi_{AB_\mathrm{com},\textsc{c}}^{A, (x,y)} \otimes \Pi_{BA_\mathrm{com},\textsc{c}}^{B, (x,y)} \right) \omega^{(x,y)}_ {AA_\mathrm{com}BB_\mathrm{com}}}.
\end{equation}
Note that attackers need to mimic the loss rate of the honest prover, so the rate of $\perp$ responses must be $1-\eta$:

\begin{equation}\label{eq:mimic_eta_P}
    \frac{1}{2^{2n} }\sum_{x,y}\tr{ \left( \Pi_{AB_\mathrm{com},\perp}^{A, (x,y)} \otimes \Pi_{BA_\mathrm{com},\perp}^{B, (x,y)} \right) \omega^{(x,y)}_ {AA_\mathrm{com}BB_\mathrm{com}}} =1-\eta.
\end{equation}

For our security proof, we limit our attention to protocols that are still secure if the quantum information travels slower than the speed of light. This motivates the following definition:
\begin{definition}\label{def:state_indep} \emph{(State-independent protocol)}. We say that a QPV protocol $\mathsf P$ is \emph{state-independent} if the protocol remains secure independently of the state $\sigma$ that the attackers pre-share at the start of the protocol.
\end{definition}

\QPVBBf~is a state-independent protocol, since it remains secure for any starting state $\sigma$ whose number of qubits is linearly bounded (in $n$) \cite{bluhm2022single}.

\subsubsection*{General structure of an attack on \textsf{c-}$\boldsymbol{\mathsf{P}}$}
In a general attack on a \textsf{c-}QPV protocol where the quantum information is sent beforehand, Alice and Bob act as follows.

\begin{enumerate}
    \item The attackers intercept the quantum information from their respective closest verifiers. Since the quantum information is sent beforehand they can apply any joint operation on it and distribute the quantum information as they wish. Note that they can also distribute entanglement in this step.
    \item Alice and Bob intercept $x$ and $y$, make a copy and send it to the other attacker, respectively. Due to relativistic constraints, they have to commit before they receive the classical information from the other party. Alice and Bob apply local quantum instruments $\{\mathcal{I}^{A}_{c_A|x}\}_{c_A\in\{0,1\}}$ and $\{\mathcal{I}^{B}_{c_B|y}\}_{c_B\in\{0,1\}}$ on their registers of $\rho$ to determine the commitments $c_A$ and $c_B$, respectively. They send off the commitments $c_A, c_B$ to their respective verifier at the appropriate time. If $c_A=c_B=0$ or $c_A\neq c_B$, no further action is required, since in the first case the verifiers do not expect any more answers, and in the second case, the protocol is aborted. For $c_A=1$ and $c_B=1$, which will be the case from now on, Alice and Bob will use the  post-selected state $\Tilde{\mathcal{I}}_{1}^{xy}(\rho)=\mathcal{I}_1^{xy}(\rho)/\tr{\mathcal{I}_1^{xy}(\rho)}$, where  $\mathcal{I}_1^{xy}=\mathcal{I}^{A}_{1|x} \otimes \mathcal{I}^{B}_{1|y}$. Alice can send a share of her state to Bob and vice versa.
    \item Upon receiving the information sent by the other party, each attacker can again locally apply an arbitrary quantum channel depending on $(x,y)$, followed by local POVMs on the state they share to obtain classical answers which will be sent to $V_0$ and $V_1$, respectively, if $c_A=1$ and $c_B=1$. Similarly to before, define a partition $AA_\mathrm{com}BB_\mathrm{com}$ and denote the final state on which they measure by $\omega^{\mathcal{I}_1, (x,y)}$.
\end{enumerate}
The attack structure is depicted in Figure~\ref{fig:attack-cBB84f}. Then the probability that the attackers answer the correct values to the verifiers is given by
\begin{equation}\label{eq prob attack with instruments}
   \mathbb P\left[\mathrm{attack }\hspace{1mm} \textsf{c-}\mathsf{P}_{\eta_V, \eta_P}\right]=\frac{1}{\eta_P}\frac{1}{2^{2n} }\sum_{x,y}\tr{ \left( \Pi_{AB_\mathrm{com},\textsc{c}}^{A, (x,y)} \otimes \Pi_{BA_\mathrm{com},\textsc{c}}^{B, (x,y)} \right) \omega^{\mathcal{I}_1, (x,y)}_ {AA_\mathrm{com}BB_\mathrm{com}}}.
\end{equation}
Here, the attackers need to mimic the transmission rate of the prover's laboratory $\eta_P$ in the rounds they commit to play, i.e. 
\begin{equation}\label{eq:mimic_eta_P_commitment}
    \frac{1}{2^{2n} }\sum_{x,y}\tr{ \left( \Pi_{AB_\mathrm{com},\perp}^{A, (x,y)} \otimes \Pi_{BA_\mathrm{com},\perp}^{B, (x,y)} \right) \omega^{\mathcal{I}_1, (x,y)}_ {AA_\mathrm{com}BB_\mathrm{com}}}= 1-\eta_P.
\end{equation}

\begin{figure}[htbp]
    \centering
    \begin{tikzpicture}[node distance=3cm, auto]
    \node (V0) {$V_0$};   
    \node [left=2cm of V0](t0){};
    \node [below=8cm of t0](t1){};
    \draw [->] (t0) -- (t1) node[midway] {time};

    \node [right=1cm of V0] (A) {$A$};
    \node [right= 5cm  of A] (B) {$B$};
    \node [right=1cm of B] (V1) {$V_1$};
    \node [below =0.5cm of V0](below_V0){};
    \node [below =0.5cm of V1](below_V1){};
    \node [right =4cm of below_V0](middle){$\rho$};
    \node [above =0.25cm of middle](middle0){};

    \node [below =0.1cm of middle](middle0){$\Tilde{\mathcal{I}}_{1}^{xy}(\rho)$};
    \node [below =0.5cm of A](A_intercepts_Q){};
    \node [below =0.5cm of B](B_intercepts_Q){};
    \node [below =0.5cm of A](A_intercepts_x){$\mathcal{I}^{A}_{c_A|x}$};
    \node [below =0.5cm of B](B_intercepts_y){$\mathcal{I}^{B}_{c_B|y}$};

    \node [above=0.5cm of A] (posV00) {};
    \node [left=2.5cm of posV00] (posV0) {};
    \node [above=0.5cm of B] (posV11) {};
    \node [right=2.5cm of posV11] (posV1) {};
    \draw [->] (posV0) -- (posV1) node[midway] {position};

    \node [below =5cm of A](A_commits){};
    \node [below =5cm of B](B_commits){};
    \node [below =6cm of A](A_answers){};
    \node [below =0.1cm of A_answers, ](A_POVM){$\{ \Pi^{A,(x,y)}_{AB_\text{com},a}\}$};
    \node [below =6cm of B](B_answers){};
    \node [below =0.1cm of B_answers](B_POVM){$\{ \Pi^{B,(x,y)}_{A_\text{com}B,b}\}$};

    \node [left=0.3cm of A](leftA){};
    \node[below=5.6cm of leftA](cA){$c_A$};
    \node [left=0.2cm of A](leftA2){};
    \node[below=6.6cm of leftA2](aA){};

    \node [right=0.3cm of B](leftB){};
    \node[below=5.6cm of leftB](cB){$c_B$};
    \node [right=0.2cm of B](leftB2){};
    \node[below=6.6cm of leftB2](bB){};
    
    \node [below =6.2cm of V0](V0bis){};
    \node [below =6.2cm of V1](V1bis){};
    \node [below =7.2cm of V0](V0bis1){};
    \node [below =7.2cm of V1](V1bis1){};

    \draw [->, transform canvas={xshift=0pt, yshift = 0 pt}, quantum] (A_intercepts_x) -- (B_answers) node[midway] (x) {};

    \draw [->, transform canvas={xshift=0pt, yshift = 0 pt}, quantum] (B_intercepts_y) -- (A_answers) node[midway] (x) {};

    \draw [->] (A_intercepts_x) -- (A_commits) {};
    \draw [->] (B_intercepts_y) -- (B_commits) {};
    \draw [->] (A_commits) -- (cA) {};
    \draw [->] (A_answers) -- (aA) {};
    \draw [->] (B_commits) -- (cB) {};
    \draw [->] (B_answers) -- (bB) {};

    \draw[dotted][-](A_intercepts_x)--(B_intercepts_y){};
    \draw[dotted][-](A_answers)--(B_answers){};
    \node [right =2.5cm of A_answers](middle2){};
    \node [below =4.7cm of middle](middle3){$\omega^{\mathcal{I}_1, (x,y)}$};
\end{tikzpicture}
\caption{Schematic representation of a general attack on a \textsf{c-}QPV protocol, where straight lines represent classical information, and undulated lines represent quantum information, including $x$ and $y$.}
\label{fig:attack-cBB84f}
\end{figure}
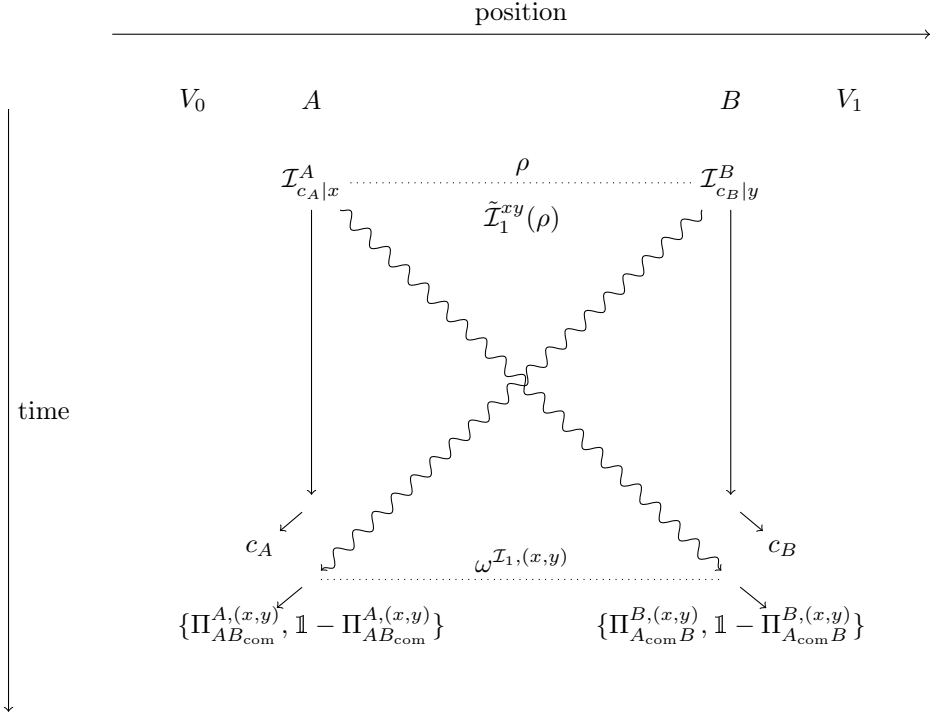~

\subsection{Security proof}\label{sec:Sec_proof}

We move on to prove the security of c-QPV. The idea is to reduce the security of a protocol with commitment \textsf{c-}$\mathsf{P}_{\eta_V, \eta_P}$ to the one of the underlying protocol without commitment $\mathsf{P}_{\eta_P}$ and (much larger) transmission rate $\eta_P$ with $\eta_V$ becoming irrelevant. The intuition is as follows. We decompose the respective quantum instruments of the attackers into a POVM measurement followed by a quantum channel. This can in fact always be done, see, e.g., Thm.\ 7.2 in \cite{hayashi2016quantum} or Lemma \ref{lemma: instrument and channel}. The respective measurement outcomes correspond to the attacker committing to play the protocol or not. To not be caught, both commitments have to be equal. We use this fact to our advantage and show, using the Gentle Measurement Lemma (Lemma \ref{lemma: gentle measurement}), that the \textit{post-measurement state} must be independent from $x,y$, and can therefore be replaced by some fixed state $\tau_{AB}$. The quantum channels that the attackers apply on their post-measurement state \textit{can} depend on $x,y$. Since these operations are only applied when both attackers commit they can also be applied in attacks on the original protocol $\mathsf{P}$. Now note that if the underlying protocol $\mathsf{P}_{\eta_P}$ remains secure for any adversarial input state that is independent of $x,y$, the attackers find themselves in the same situation as attacking $\mathsf{P}_{\eta_P}$ (with input $\tau_{AB}$) when they attack \textsf{c-}$\mathsf{P}_{\eta_V, \eta_P}$. Then, the success probability of attacking \textsf{c-}$\mathsf{P}_{\eta_V, \eta_P}$ should be close to the success probability of attacking $\mathsf{P}_{\eta_P}$, and we inherit security of the underlying protocol in its committing version.

We start by showing that any quantum instrument can be decomposed into a measurement followed by a quantum channel turns out to be a crucial ingredient in our proof. We include a short proof of it for convenience.

\begin{lemma}\label{lemma: instrument and channel} \emph{(E.g.\ Thm 7.2 in \cite{hayashi2016quantum})}
Let $\mathcal{I}=\{\mathcal{I}_i\}_{i\in\Omega}$ be an instrument, and $\{M_i\}_i$ its corresponding POVM, i.e.\ $\mathcal{I}_i^\dagger(\mathbbm{1}) = M_i$. 
Then, for every $i\in\Omega$, there exists a quantum channel (CPTP map) $\mathcal{E}_i$ such that  
    \begin{equation}
        \mathcal{I}_i(\rho)=\mathcal{E}_i \left(\sqrt{M_i}\rho\sqrt{M_i}\right)
    \end{equation}
\end{lemma}
\begin{proof} Let $\{K_j\}_j$ be a Kraus decomposition of $ \mathcal{I}_i$, whose existence is guaranteed by Lemma \ref{lemma thm Kraus decomposition}. Since
\begin{equation}
    \tr{\mathcal{I}_i(\rho)} = \tr{\sum_jK_j\rho K_j^{\dagger}} = \tr{\rho \sum_jK_j^{\dagger}K_j}=\tr{\rho M_i}
\end{equation}
for any state $\rho$, we have $M_i=\sum_jK_j^{\dagger}K_j$. Denote the pseudo-inverse of $\sqrt{M_i}$ by $(\sqrt{M_i})^{-}$ and let $P$ be the orthogonal projection onto the support of $\sqrt{M_i}$, i.e.\ $P = \sqrt{M_i} \left( \sqrt{M_i} \right)^-$. Then note that
\begin{align}
    \sum_j \left( \sqrt{M_i} \right)^{-} K_j^\dagger K_j \left( \sqrt{M_i} \right)^{-} = \left( \sqrt{M_i} \right)^{-} M_i \left( \sqrt{M_i} \right)^{-} = P^\dagger P = P.
\end{align}
Hence, if we add $\mathbbm{1}-P$ on both sides, we obtain a full Kraus decomposition $\left\{K_j(\sqrt{M_i})^{-}, \mathbbm{1} - P \right\}_j$ of a map, call it $\mathcal{E}_i$, that adds up to the identity. Thus, by Lemma \ref{lemma thm Kraus decomposition}, $\mathcal{E}_i$ is completely positive and trace preserving, i.e.\ a quantum channel. Finally, we see that
\begin{align}
    \mathcal{E}_i\left(\sqrt{M_i}\rho\sqrt{M_i}\right)&= (\mathbbm{1}-P) \sqrt{M_i}\rho\sqrt{M_i} (\mathbbm{1}-P) + \sum_j K_j(\sqrt{M_i})^{-}\sqrt{M_i}\rho\sqrt{M_i} (\sqrt{M_i})^{-} K_j^{\dagger} \nonumber \\
    &=\sum_jK_j\rho K_j^{\dagger}=\mathcal{I}_i(\rho), 
\end{align}
as desired. The last equation follows from the fact that $(\mathbbm{1}-P)\sqrt{M_i} = \sqrt{M_i} - \sqrt{M_i}(\sqrt{M_i}^{-})\sqrt{M_i} = 0$, which is one of the defining properties of the pseudo-inverse, and that $K_j P = K_j$. This follows in turn via $M_i = \sum_j K_j^\dagger K_j$, implying that $\mathrm{ker}(M_i) \subseteq \mathrm{ker}(K_j)$ for all $j$. In other words, $\mathrm{supp}(K_j) \subseteq \mathrm{supp}(M_i)=\mathrm{supp}(\sqrt{M_i})$ for all $j$, and $P$ projects onto the latter. Hence $K_j P = K_j$.
\end{proof}
Combining the Stinespring dilation with Lemma~\ref{lemma: instrument and channel} allows us to see the operations of the attackers after the commit-measurement as a unitary in a larger space, and yields the following decomposition of quantum instruments. 

\begin{corollary}\label{corollary: instrument and stinespring} Let $\mathcal{I}=\{\mathcal{I}_i\}_{i\in\Omega}$ be an instrument, and $\{M_i\}_{i \in \Omega}$ its corresponding POVM. 
Then, for every $i\in\Omega$,  there exists an environment Hilbert space $\mathcal{H}_E$ and a unitary $U_i$ on $\mathcal{H}\otimes\mathcal{H}_E$ such that
    \begin{equation}\label{eq unitary extension intrument}
        \mathcal{I}_i(\rho)=\ptr{E}{U_i \left( \sqrt{M_i} \rho \sqrt{M_i} \otimes \ketbra{0}{0}_E \right) U_i^{\dagger}}
    \end{equation}
     for all $\rho\in\mathcal{B}(\mathcal{H})$,
\end{corollary}

In the case of a commit round of a QPV protocol the subscript denotes whether the attackers commit ($i = 1$) or do not commit ($i=0$). The unitary $U_i$ in eq.~\eqref{eq unitary extension intrument} is the unitary corresponding to a Stinespring dilation of the channel $\mathcal{E}_i$ appearing in Lemma \ref{lemma: instrument and channel}. We denote the POVMs corresponding to the instruments $\{\mathcal{I}^{A}_{c_A|x}\}_{c_A}$ and $\{\mathcal{I}^{B}_{c_B|y}\}_{c_B}$ of Alice and Bob by $\left\{M^x_{A}, \mathbbm{1}-M^x_{A}\right\}$ and $\left\{M^y_{B},\mathbbm{1}-M^y_{B} \right\}$ respectively. Here the POVM elements $M^x_{A}$ and $M^y_{B}$ correspond to the measurement outcome `commit' ($c_A=1$ and $c_B=1$). We denote the post-measurement state, given by the above POVMs, corresponding to Alice and Bob committing, for input $x,y$, by:
    \begin{equation}\label{eq rhoxy}            
    \rho^{xy}:=\frac{\left( \sqrt{M^x_A} \otimes \sqrt{M^y_B} \right) \, \rho \, \left( \sqrt{M^x_A} \otimes \sqrt{M^y_B} \right)}{\tr{\left( M^x_A \otimes M^y_B \right) \rho}}.
    \end{equation}
Since the attackers only continue to play if they both commit, we know that the unitary that they apply on their post measurement state on inputs $x,y$ is $U_{1,A}^x \otimes U_{1,B}^y$. 

We will show that $\rho^{xy}$ and $\rho^{x'y'}$, for arbitrary $(x,y)$ and $(x',y')$ cannot differ too much from each other. First we recall the \textit{Gentle Measurement Lemma}:

\begin{lemma} \label{lemma: gentle measurement}\emph{(Gentle Measurement Lemma \cite{winter1999coding})}  Let $\rho$ be a quantum state and $\{ M, \mathbbm{1} - M \}$ be a two-outcome measurement. If $\tr{M \rho} \geq 1-\varepsilon$, then the post-measurement state \begin{align}    \rho'=\frac{\sqrt{M}\rho\sqrt{M}}{\tr{ M\rho}}
\end{align}
of measuring $M$ fulfills
\begin{align}
||\rho-\rho'||_1\leq 2\sqrt{\varepsilon}. \end{align}
\end{lemma}

To avoid detection, the commitments \( c_A \) and \( c_B \) must be equal, or at least indistinguishable with very high probability. This implies that if Alice has measurement outcome $c_A = 1$ Bob's measurement is already determined, and vice-versa. This means, by the Gentle Measurement Lemma, that Bob's measurement cannot have disturbed the state too much. The following lemma relates the closeness of states to the probability of answering different commits, given that one party commits.

\begin{lemma}\label{lemma: paths between strings}\emph{(Paths Between Strings)} 
Assume that for inputs $(x,y)$, $(x',y)$ and $(x',y')$ in $\{0,1\}^{2n}$ that the probability that one party does not commit, given that the other party commits, is upper bounded by some $\varepsilon > 0$. Then,
\begin{align}
    \| \rho^{xy}-\rho^{x'y'}\|_1 \leq 8 \sqrt{\varepsilon}.
\end{align}
\end{lemma}

\begin{proof}
Consider the attackers Alice and Bob performing the most general attack described above and the POVMs $\{M^x_{A}, \mathbbm{1}-M^x_{A}\}$ and $\{M^y_{B},\mathbbm{1}-M^y_{B} \}$ as defined above. We write 
    \begin{equation}\label{eq rhox, rhoy, rhoxy}
    \rho^{x,(\cdot)}=\frac{( \sqrt{M^x_A} \otimes \mathbbm{1}_B ) \, \rho \, ( \sqrt{M^x_A} \otimes \mathbbm{1}_B )}{\tr{( M^x_A \otimes \mathbbm{1}_B )\rho}}, \quad \quad \rho^{(\cdot),y}=\frac{(\mathbbm{1}_A \otimes \sqrt{M^y_B}) \, \rho \, (\mathbbm{1}_A \otimes \sqrt{M^y_B})}{\tr{(\mathbbm{1}_A \otimes M^y_B)\rho}}
    \end{equation}
for the post measurement states corresponding to only Alice or Bob committing before applying the quantum channel. By assumption, we have:
    \begin{align}\label{equ:diffcommiteps}
         \tr{\left( \mathbbm{1}_A \otimes (\mathbbm 1-M^y_{B}) \right) \rho^{x,(\cdot)}} \leq \varepsilon, \qquad \tr{\left( (\mathbbm 1-M^x_{A}) \otimes \mathbbm{1}_B \right) \rho^{(\cdot),y}}\leq\varepsilon.
    \end{align}
Similarly for the input $(x',y)$ and $(x',y')$ we get:
    \begin{align}
         &\tr{\left( \mathbbm{1}_A \otimes (\mathbbm 1-M^y_{B}) \right) \rho^{x',(\cdot)}} \leq \varepsilon, &&\tr{\left( (\mathbbm 1-M^{x'}_{A}) \otimes \mathbbm{1}_B \right) \rho^{(\cdot),y}}\leq\varepsilon, \\
         &\tr{\left( \mathbbm{1}_A \otimes (\mathbbm 1-M^{y'}_{B}) \right) \rho^{x',(\cdot)}} \leq \varepsilon, &&\tr{\left( (\mathbbm 1-M^{x'}_{A}) \otimes \mathbbm{1}_B \right) \rho^{(\cdot),y'}}\leq\varepsilon.
    \end{align}
Therefore, by  Lemma \ref{lemma: gentle measurement} (Gentle Measurement Lemma) we get the following inequalities: 
\begin{equation}\label{equation rhox-rhoxy}
\begin{aligned}
    &\norm{\rho^{(\cdot),y}-\rho^{xy}}_{1}\leq2\sqrt{\varepsilon}, &&\norm{\rho^{(\cdot),y}-\rho^{x'y}}_{1}\leq2\sqrt{\varepsilon} \\
    &\norm{\rho^{x',(\cdot)}-\rho^{x'y}}_{1}\leq2\sqrt{\varepsilon},
    &&\norm{\rho^{x',(\cdot)}-\rho^{x'y'}}_{1}\leq2\sqrt{\varepsilon}    
\end{aligned}
\end{equation}
Now we get for the trace distance between the two density matrices:
\begin{equation}\label{eq rh00-rhoxy}
\begin{split}
    \norm{\rho^{x'y'}-\rho^{xy}}_1&=\norm{\rho^{x'y'}-\rho^{x',(\cdot)}+\rho^{x',(\cdot)}-\rho^{x'y}+\rho^{x'y}-\rho^{(\cdot),y}+\rho^{(\cdot),y}-\rho^{xy}}_1\\
    &\leq \norm{\rho^{x'y'}-\rho^{x',(\cdot)}}_1+\norm{\rho^{x',(\cdot)}-\rho^{x'y}}_1+\norm{\rho^{x'y}-\rho^{(\cdot),y}}_1+\norm{\rho^{(\cdot),y}-\rho^{xy}}_1 \\
    &\leq8\sqrt{\varepsilon},
    \end{split}
\end{equation}
where we used the triangle inequality and eq.~\eqref{equation rhox-rhoxy}.
\end{proof}

Note that if the probability of answering different commits on the inputs $(x,y')$ instead of $(x',y)$ was small we would get the same inequality between $\rho^{xy}$ and $\rho^{x'y'}$.

In general, an honest prover will never answer different commit bits back to the verifiers. Thus one could argue that the probability of answering `no commit' when the other party answers `commit' should be zero. In that case, by Lemma \ref{lemma: paths between strings}, we see that all post-measurement states are equal, and thus independent of $x,y$. Then, the quantum instrument that Alice and Bob apply adds no extra power and their actions are contained in the actions they could do in attacking a state-independent protocol (cf. Definition~\ref{def:state_indep}). And the probability to attack the protocol successfully on rounds in which the attackers commit is equal to the original protocol.     This is summarized in the following corollary:
\begin{corollary}
    If we demand perfect coordination for the commitments in attack strategies, then for any state-independent quantum position verification $\mathsf{P}$ its version with commitment \emph{\textsf{c-}}$\mathsf{P}$ is fully loss tolerant against transmission loss. That is,
    \begin{align}
        \mathbb P[\mathrm{attack } \, \emph{\textsf{c-}}\mathsf{P}_{\eta_{V},\eta_P}] = \mathbb P[\mathrm{attack } \, \mathsf P_{\eta_{P}}].
    \end{align}
    Thus, protocols like \QPVBBf~can be made secure against transmission loss.
\end{corollary}

However, one can argue setting the probability to answer `no commit' given that the other party answers `commit' to zero is too restrictive. Also when this probability is sufficiently low, with high probability the attackers will not get detected by answering different commitments. Thus, it could be that this strategy outperforms the original attack strategy. This stronger setting is not always considered in QPV protocols, but is relevant in practice. We will show that allowing for this possibility does not help the attackers much, and we can still show security. We give a continuity statement on the probability of attacking successfully, showing that the protocols with a commitment round perform similarly to the original protocol if the probability of answering different commitments is small. Again the proof strategy is to show that the post-measurement states must be close to each other, depending on the probability of committing differently, given that one party commits (the rounds in which no-one commits are discarded).

The statement of Lemma \ref{lemma: paths between strings} can be formulated as a connection problem in a graph. The local inputs $x,y$ are represented as vertices in a bipartite graph, and we connect two vertices $x,y$ if the probability that the two parties send different commitments is upper bounded by $\varepsilon$ as in the proof of the above lemma. Then for two pairs of inputs $x,y$ and $x',y'$ (i.e.~edges in the graph) $\| \rho^{xy} - \rho^{x'y'} \|_1 \leq 8 \sqrt{\varepsilon}$, if there is an edge in the graph that connects either $x',y$ or $x,y'$. This is represented in Figure \ref{fig:bipartite_all_edges}.

\begin{figure}[htbp]
\centering
\begin{tikzpicture}[scale=1.5, line width=10pt]
  
  \foreach \i in {0,...,4}
    \foreach \j in {0,...,4}
      \draw[line width=1pt, opacity=0.8] (0,\i) -- (2,\j);
      
  \draw[red, line width=4pt, opacity=0.7] (0,3) -- (2,2);
  
  \draw[orange, line width=4pt, opacity=0.7] (0,1) -- (2,2);
  
  \draw[green, line width=4pt, opacity=0.7] (0,1) -- (2,0);

  \foreach \y/\label in {0/, 1/$x'$, 2/, 3/$x$, 4/}
    \fill (0,\y) circle (2pt) node[left, xshift=-2pt] {\label};
  
  \foreach \y/\label in {0/$y'$, 1/, 2/$y$, 3/, 4/}
    \fill (2,\y) circle (2pt) node[right] {\label};

\end{tikzpicture}
\caption{Graphical representation of converting the pair $(x,y)$ (red) to $(x',y')$ (green) via $(x',y)$ (orange). Vertices on the left correspond to possible inputs $x$, on the right to possible inputs $y$. A connection between two strings means that the probability of committing differently on this input is smaller than $\varepsilon$.}
\label{fig:bipartite_all_edges}
\end{figure}
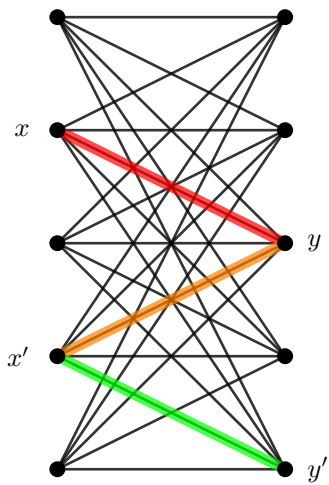

Importantly, the statement of Lemma \ref{lemma: paths between strings} only holds if the probability of answering different commit bits, given that one party commits, is upper bounded by $\varepsilon$ for all three pairs of strings. However, this is not something that the verifiers can enforce to be true for every pair of strings. The verifiers can only check for the rounds that they play whether the commitments are equal, but given that there are $2^{2n}$ possible inputs they cannot get the commit statistics for all of them.

It could be that allowing the attackers to commit differently on a subset of strings can outperform attackers that have to behave well over all strings. Since this subset is unknown to the verifiers (as it is part of the attack strategy) the probability to detect a wrong commit can be made as small as the relative size of the subset to the total set.

We can visualize the problem of committing differently intuitively via the complete bipartite graph in Figure~\ref{fig:bipartite_all_edges}. In the figure, two vertices are connected if the probability of answering different commitments is upper bounded by $\varepsilon$. Allowing attackers to answer different commits with a higher probability is equivalent to removing certain edges in this graph.

We then still have a bipartite graph but not all edges are connected. What we are now interested in is how many edges can still be reached within two steps from some other edge. It turns out that even if we allow attackers to commit differently with probability higher than $\varepsilon$ on a constant fraction of edges, there will be an edge that will be connected to at least a constant fraction of other edges in two steps (as used in Lemma~\ref{lemma: paths between strings}). 

\begin{lemma}[Edge Removal]\label{lemma: edge removal}
    Consider a complete bipartite graph whose independent sets are of equal size $2^n$. After removing a constant fraction $\Tilde{c}\leq\frac{1}{2}$ of edges, there exists an edge such that the number of edges that can be reached from this edge in two steps is at least $(1 - 2\Tilde{c}) 2^{2n}$.
\end{lemma}
\begin{proof}
    The number of edges of a complete bipartite graph with $2^n$ nodes in its independent sets is $2^{2n}$, as there are $2^n$ edges for any vertex. Now suppose we remove $\Tilde c \cdot 2^{2n}$ of these edges. Then, there must be a vertex $l$ on the left with at least $(1-\Tilde c) 2^n$ connecting edges. Let one of these edges be your starting edge. Now consider all the vertices on the right that are connected to $l$. Before we removed any edges there were $2^n$ edges connecting each of these vertices to the left. However, we removed $\Tilde c \cdot 2^{2n}$ of these edges, so the number of edges going back is now at least $(1-\Tilde c) \cdot 2^{2n} - \Tilde c \cdot 2^{2n} = (1 - 2\Tilde c) 2^{2n}$. Thus there are $(1 - 2\Tilde c) 2^{2n}$ edges that can be reached in two steps from the starting edge.
\end{proof}
Now let us split up the set of all possible inputs into one set where the probability of not committing, given that the other party commits, is lower than $\varepsilon$ and its complement. We write 
\begin{equation}\label{equ:Sigma_eps}
    \Sigma_{\varepsilon}:=\{x,y\mid \tr{\left( \mathbbm{1} \otimes (\mathbbm{1}-M^y_B) \right) \rho^{x,(.)}} \leq\varepsilon \wedge  \tr{\left( \mathbbm{1}-M^x_A) \otimes \mathbbm{1} \right) \rho^{(.),y}}\leq\varepsilon \},
\end{equation}
which can also be written in terms of conditional probabilities
\begin{equation}
    \Sigma_\varepsilon = \{x,y \mid \mathbb{P}[c_B = 0 \mid c_A = 1, x_A,y_B] \leq \varepsilon \wedge \mathbb{P}[c_A = 0 \mid c_B = 1, x_A, y_B] \leq \varepsilon\},
\end{equation}
where the subscript $A,B$ denote that the information about the strings $x,y$ is only known to player $A$ or $B$ and not both. We denote by $\Sigma^c_{\varepsilon}$ the complementary set of $\Sigma_{\varepsilon}$.
Using this definition we can show the following.

\begin{lemma} \label{lemma:existence_reference_state_close_to_other_states} If $\abs{\Sigma^c_{\varepsilon}}\leq \Tilde{c} 2^{2n}$, then there is a pair $(x^*,y^*)$ such that there exist at least  $ (1 - 2\Tilde{c}) 2^{2n}$ pairs  $(x',y')\in\Sigma_{\varepsilon}$ fulfilling
    \begin{align}
    \| \rho^{x^*y^*}-\rho^{x'y'}\|_1 \leq 8 \sqrt{\varepsilon}.
\end{align}
\end{lemma}

\begin{proof}
$\abs{\Sigma^c_{\varepsilon}}\leq \Tilde{c} 2^{2n}$, so at most there are a fraction of $\Tilde{c}$ edges removed from the complete bipartite graph. By Lemma \ref{lemma: edge removal} there is a pair $(x^*,y^*)$ from which there are at least $(1-2\Tilde{c})2^{2n}$ edges connected in two steps. Applying Lemma \ref{lemma: paths between strings} gives the desired statement.
\end{proof}

Furthermore, even though the honest prover will commit to playing, it is still possible for him to answer loss due to his errors in his measurement setup. We bound the probability of success from above by showing that the commitment procedure of the attackers cannot change the input state much depending on $x,y$. In turn this means that there is a state $\rho^{x^*y^*}$ that is close to \textit{most} states. It also means that the probability for attackers to answer loss on $\rho^{x^*y^*}$ cannot differ too much from the overall loss rate $\eta_P$. We make this precise in the following lemmas. 

We will start with the following technical lemma:
\begin{lemma} \label{lemma:taumonotone}
    Let the function $f$ be defined as
\begin{align}
    f(x,\delta,\eta) = \frac{\sqrt{(\delta +\eta -x)^2+4 \delta  x}+ x -\delta -\eta }{2 x}.
\end{align}  
Then, $f$ is monotonically non-decreasing in the first argument $x$ in the regime of $x$, $\delta$, $\eta \geq 0$.
\end{lemma}
\begin{proof}
    We need to show that for $x \geq 0$, $\frac{\partial f(x,\delta,\eta)}{\partial x} \geq 0$. We start by calculating
    \begin{align}
        \frac{\partial f(x,\delta,\eta)}{\partial x} = \frac{(\delta +\eta ) \left(-\delta -\eta +\sqrt{\eta ^2+2 \eta  (\delta -x)+(\delta +x)^2}\right)+x (\eta -\delta )}{2 x^2 \sqrt{(\delta +\eta -x)^2+4 \delta  x}}.
    \end{align}
    Since the denominator is always positive, we need to show the numerator is positive. To do so, note that the numerator is $0$ in $x=0$ for any $\delta, \eta \geq 0$. So if we prove that the numerator is also monotonically non-decreasing in $x$ we get that $f$ is monotonically non-decreasing in $x$. We write:
    \begin{align}
        g(x, \delta, \eta) &= (\delta +\eta ) \left(-\delta -\eta +\sqrt{\eta ^2+2 \eta  (\delta -x)+(\delta +x)^2}\right)+x (\eta -\delta ), \\ 
        \frac{\partial g(x,\delta,\eta)}{\partial x} &= \frac{(\eta -\delta ) \sqrt{4 \delta \eta +(\delta -\eta +x)^2}+(\delta +\eta ) (\delta -\eta +x)}{\sqrt{4 \delta \eta +(\delta -\eta +x)^2}}.
    \end{align}
    Note that the denominator is always positive and $0$ in $x=0$ for any $\delta, \eta \geq 0$. Again, if we show that the numerator is monotonically non-decreasing in $x$, then we show that $g$ is positive:
    \begin{align}
        h(x, \delta, \eta) &= (\eta -\delta ) \sqrt{4 \delta \eta +(\delta -\eta +x)^2}+(\delta +\eta ) (\delta -\eta +x), \\
        \frac{\partial h(x,\delta,\eta)}{\partial x} &= \frac{(\delta +\eta ) \sqrt{4 \delta \eta +(\delta -\eta +x)^2}-(\delta -\eta ) (\delta -\eta +x)}{\sqrt{4 \delta \eta +(\delta -\eta +x)^2}}.
    \end{align}
    Now note $\sqrt{4 \delta \eta +(\delta -\eta +x)^2} \geq \abs{\delta -\eta +x}$, as $\delta, \eta \geq 0$. Then 
    \begin{gather}
        (\delta +\eta ) \sqrt{(\delta +\eta )^2+x^2+2 x (\delta -\eta )}-(\delta -\eta ) (\delta -\eta +x) \geq (\delta +\eta) |\delta -\eta +x|-(\delta -\eta ) (\delta -\eta +x) \nonumber \\
        =\delta (|\delta -\eta +x|-(\delta -\eta +x))+\eta (|\delta -\eta +x|+(\delta -\eta +x)).
    \end{gather} 
    Finally, $|\delta -\eta +x|-(\delta -\eta +x)$ and $|\delta -\eta +x|+(\delta -\eta +x)$ are both positive by definition and we have that $h$ is monotonically non-decreasing and $0$ in $x=0$. Thus $h$ is positive everywhere, thus $g$ is positive everywhere, thus $\frac{\partial f(x,\delta,\eta)}{\partial x}$ is positive and $f$ is monotonically non-decreasing.
\end{proof}

After creating the commitment bit both attackers exchange a quantum system and apply some measurement on their systems. If we fix a partition into systems $AA_\text{com}BB_\text{com}$, where `com' denotes the subsystems that will be communicated. We can write the attackers' POVMs as $\{ \Pi^{A,(x,y)}_{AB_\text{com},a}\}_{a\in\{0,1,\perp\}}$ and $\{ \Pi^{B,(x,y)}_{A_\text{com}B,b}\}_{b\in\{0,1,\perp\}}$, where we denote `correct' as $0$, `incorrect' as $1$, and `no-photon' as $\perp$ answers. For simplicity we will write the POVM of answering correctly as $ \Pi^{xy} \coloneqq \left( \Pi^{A,(x,y)}_{A,0} \otimes \Pi^{B,(x,y)}_{B,0}\right)$. Notice that these POVMs have to fulfill \eqref{eq:mimic_eta_P_commitment}. We want to estimate the probability of answering correctly on $\rho^{x^*y^*}$:
\begin{align} \label{eq:estimatebyPetaP}
    \frac{1}{2^{2n} \eta_P} \sum_{x,y} \tr{ \Pi^{xy} \mathcal{E}_1^{xy}( \rho^{x^*y^*} )}, 
\end{align}
by $\mathbb{P}[\text{attack} \mathsf{P}_{\eta_P}]$. Thus, we want to enforce that the probability of answering loss is equal to the loss rate of the underlying protocol. However the loss rate could be different for some particular $\rho^{x^*y^*}$ that we get from Lemma \ref{lemma:existence_reference_state_close_to_other_states}. On the other hand, the loss rate cannot differ too much as on average the loss rate is equal to $1-\eta_P$ and $\rho^{x^*y^*}$ is close to \textit{most} other states. We will show that there exist modified POVMs for attackers $A,B$ such that the correct loss rate is indeed attained, allowing us to upper bound eq. \eqref{eq:estimatebyPetaP} by $\mathbb{P}[\text{attack } \mathsf{P}_{\eta_P}]$, at the cost of a less tight bound. The following lemma makes this precise. We first define $\Lambda^{(x,y)}_\varepsilon$ to be the set of all quantum states close to some reference state $\rho^{xy}$:
\begin{equation}
	\Lambda^{(x,y)}_\varepsilon \coloneq \left\{ (x',y')  \in \Sigma_\varepsilon : \| \rho^{xy} - \rho^{x'y'} \|_1 \leq 8 \sqrt{\varepsilon} \right\},
\end{equation}
and write $\Lambda_\varepsilon \coloneq \Lambda^{(x_*,y_*)}_\varepsilon$.

For the next lemma we define the following quantities $\Delta_\varepsilon \coloneq \frac{|\Lambda_\varepsilon|}{2^{2n}} 8\sqrt{\varepsilon} + \frac{2|\Lambda^c_\varepsilon|}{2^{2n}}$ and $g(\eta_P)\coloneq \max\left(\frac{1}{ \sqrt{\eta_P}},\frac{1}{(1-\eta_P)}\right)$. Typically in our proofs we will have $\Delta_\varepsilon$ close to $0$, so $\Delta_\varepsilon \leq \sqrt{\Delta_\varepsilon}$.

\begin{lemma} \label{lemma:upperboundforrhoxstarystar}
    If $\mathsf P$ is state-independent (cf.~Definition~\ref{def:state_indep}) then and we assume $\Delta_\varepsilon \leq \sqrt{\Delta_\varepsilon}$, for any state $\rho^{x^*y^*}$ we have:
    \begin{align}
         \frac{1}{2^{2n} \eta_P} &\sum_{x,y \in \Lambda_\varepsilon} \tr{ \Pi^{xy} \mathcal{E}_1^{xy}(\rho^{x^*y^*})} \leq \mathbb{P}[\mathrm{attack } \mathsf{P}_{\eta_P}] + \sqrt{\Delta_\varepsilon} g(\eta_P) \frac{|\Lambda_\varepsilon|}{2^{2n} \eta_P}.
    \end{align}
\end{lemma}

\begin{proof}
We start with the probability of answering loss for the specific state $\rho^{x^*y^*}$:
\begin{align}
    \mathbb{P}[\perp(\rho^{x^*y^*})] = \frac{1}{2^{2n}} \sum_{x,y} \tr{ \left(\Pi_\perp^{A,x} \otimes \Pi_\perp^{B,y}\right) \mathcal{E}_1^{xy}(\rho^{x^*y^*})} = 1 - \eta_P + \delta,
\end{align}
where $\delta$ specifies how close the loss rate is to the average loss rate. We can bound $\delta$ as follows:
\begin{align}
    \mathbb{P}[\perp(\rho^{x^*y^*})] &- \frac{1}{2^{2n}} \sum_{x,y} \tr{ \left(\Pi_\perp^{A,x} \otimes \Pi_\perp^{B,y}\right) \mathcal{E}_1^{xy}(\rho^{xy})} + \frac{1}{2^{2n}} \sum_{x,y} \tr{ \left(\Pi_\perp^{A,x} \otimes \Pi_\perp^{B,y}\right) \mathcal{E}_1^{xy}(\rho^{xy})} \\
    &= \frac{1}{2^{2n}} \sum_{x,y} \tr{ \left(\Pi_\perp^{A,x} \otimes \Pi_\perp^{B,y}\right) \mathcal{E}_1^{xy}(\rho^{x^*y^*} - \rho^{xy})} + 1 - \eta_P \\
    &\leq 1 - \eta_P + \frac{1}{2^{2n}} \sum_{x,y} \| \left(\Pi_\perp^{A,x} \otimes \Pi_\perp^{B,y}\right) \|_{\infty} \| \mathcal{E}_1^{xy}(\rho^{x^*y^*} - \rho^{xy}) \|_1 \\
    &\leq 1- \eta_P + \frac{1}{2^{2n}} \sum_{x,y \in \Lambda_\varepsilon} \| \mathcal{E}_1^{xy}(\rho^{x^*y^*} - \rho^{xy}) \|_1 + \frac{1}{2^{2n}} \sum_{x,y \in \Lambda^c_\varepsilon} \| \mathcal{E}_1^{xy}(\rho^{x^*y^*} - \rho^{xy}) \|_1 \\
    &\leq 1 - \eta_P + \frac{|\Lambda_\varepsilon|}{2^{2n}} 8 \sqrt{\varepsilon} + \frac{2|\Lambda^c_\varepsilon|}{2^{2n}}
\end{align}
 so $\delta \in [-\frac{|\Lambda_\varepsilon|}{2^{2n}} 8 \sqrt{\varepsilon} - \frac{|\Lambda^c_\varepsilon|}{2^{2n}}, \frac{|\Lambda_\varepsilon|}{2^{2n}} 8 \sqrt{\varepsilon} + \frac{|\Lambda^c_\varepsilon|}{2^{2n}}]$. Since $\Delta_\varepsilon = \frac{|\Lambda_\varepsilon|}{2^{2n}} 8 \sqrt{\varepsilon} + \frac{2|\Lambda^c_\varepsilon|}{2^{2n}}$, we have
\begin{align}
    -\Delta_\varepsilon \leq \delta \leq \Delta_\varepsilon.
\end{align}
If $\delta = 0$, we have $\mathbb{P}[\perp(\rho^{x^*y^*})] = 1 - \eta_P$ and we are done. We will treat the two cases $\delta > 0$ and $\delta < 0$ separately. Assume $\delta > 0$, then loss is answered too often. We define a new POVM for $A$, $\tilde{\Pi}^A$, by changing some of the loss answers to a conclusive answer:
\begin{align}
    \tilde{\Pi}^{A, xy}_\perp &= \frac{1-\eta_P}{1-\eta_P+\delta} \Pi_{\perp}^{A,xy} \\
    \tilde{\Pi}^{A, xy}_0 &= \Pi_0^{A,xy} + \left( 1 - \frac{1-\eta_P}{1-\eta_P+\delta} \right) \Pi_{\perp}^{A,xy} \\
    \tilde{\Pi}^{A, xy}_1 &= \Pi_{1}^{A,xy}.
\end{align}
Clearly, if $\Pi^A$ is a valid POVM, then so is $\tilde{\Pi}^A$. The loss rate on $\rho^{x^*y^*}$ is now:
\begin{align}
    \frac{1}{2^{2n}}\sum_{x,y} \tr{ \left( \tilde{\Pi}^{A,xy}_\perp \otimes \Pi^{B,xy}_\perp \right) \mathcal{E}_1^{xy} (\rho^{x^*y^*})} = 1 - \eta_P,
\end{align}
as we required. The success probability of the strategy only differs slightly:
\begin{align}
    \bigl| \tr{\tilde{\Pi}^{xy} \mathcal{E}_1^{xy}(\rho^{x^*y^*})} &- \tr{\Pi^{xy} \mathcal{E}_1^{xy}(\rho^{x^*y^*})} \bigr | \leq \| \tilde{\Pi}^{xy} - \Pi^{xy} \|_\infty \\
    &\leq \left\| \left( 1 - \frac{1-\eta_P}{1- \eta_P + \delta} \right) \Pi_\perp^{A,xy} \otimes \Pi^{B} \right\|_\infty \leq 1 - \frac{1-\eta_P}{1-\eta_P + \delta} \\
    &\leq \frac{\delta}{1-\eta_P + \delta} \leq \frac{\delta}{1-\eta_P} \leq \frac{\Delta_\varepsilon}{1-\eta_P}.
\end{align}
Then the success probability on $\rho^{x^*y^*}$ conditioned on playing is bounded by:
\begin{align}
    \frac{1}{2^{2n} \eta_P} &\sum_{x,y \in \Lambda_\varepsilon} \tr{ \Pi^{xy} \mathcal{E}_1^{xy}(\rho^{x^*y^*})} \nonumber \\
    &= \frac{1}{2^{2n} \eta_P} \sum_{x,y \in \Lambda_\varepsilon} \tr{ (\Pi^{xy} - \tilde{\Pi}^{xy}) \mathcal{E}_1^{xy}(\rho^{x^*y^*})} + \frac{1}{2^{2n} \eta_P} \sum_{x,y \in \Lambda_\varepsilon} \tr{ \tilde{\Pi}^{xy} \mathcal{E}_1^{xy}(\rho^{x^*y^*})} \\
    &\leq \frac{|\Lambda_\varepsilon|}{2^{2n} \eta_P} \|\Pi^{xy} - \tilde{\Pi}^{xy}\|_\infty \| \mathcal{E}_1^{xy}(\rho^{x^*y^*})\|_1 + \mathbb{P}[\text{attack } \mathsf{P}_{\eta_P}]  \\
    &\leq \mathbb{P}[\text{attack } \mathsf{P}_{\eta_P}] + \frac{\Delta_\varepsilon}{\eta_P (1-\eta_P)} \frac{|\Lambda_\varepsilon|}{2^{2n}}, \label{eq:finalinequalitydeltapositive}
\end{align}
where the first inequality follows from the assumption that the protocol is secure against any input state and the fact that $U^{xy} = U^x \otimes U^y$ as $\mathcal{I}_1^{xy} = \mathcal{I}^A_{1|x} \otimes \mathcal{I}^B_{1|y}$. The latter we can neglect since the local unitaries can be absorbed into the attack strategy on the original protocol $\mathsf{P}_{\eta_V, \eta_P}$.

Now assume $\delta \leq 0$, for this case we also define a new POVM that shifts some answers to loss depending on a parameter $\mu \in [0.1]$ that we will choose later. For both $A,B$ we take the same transformation:
\begin{align}
    \tilde{\Pi}^{A/B, xy}_\perp &= {\Pi}^{A/B, xy}_\perp + \mu ({\Pi}^{A/B, xy}_0 +{\Pi}^{A/B, xy}_1) \\
    \tilde{\Pi}^{A/B, xy}_0 &= (1-\mu) {\Pi}^{A/B, xy}_0 \\
    \tilde{\Pi}^{A/B, xy}_1 &= (1-\mu) {\Pi}^{A/B, xy}_1
\end{align}
We then want to pick $\mu$ such that the probability to answer loss for the modified POVM is exactly $1-\eta_P$. For clarity we define $a_{ij}^{xy} = \tr{\left( \Pi^{A,xy}_i \otimes \Pi^{B,xy}_j \right) \mathcal{E}_1^{xy}(\rho^{x^*y^*})}$:
\begin{align}
    &\frac{1}{2^{2n}} \sum_{x,y} \tr{ \left(\tilde{\Pi}^{A,xy}_\perp \otimes \tilde{\Pi}^{B,xy}_\perp \right) \mathcal{E}_1^{xy}(\rho^{x^*y^*})} = 1 - \eta_P \label{eq:probabilityoflosspitilde}\\
    &= \frac{1}{2^{2n}} \sum_{x,y} \left( a_{\perp\perp}^{xy} + \mu (a^{x,y}_{\perp 0} + a^{x,y}_{0 \perp} + a^{x,y}_{\perp 1} + a^{x,y}_{1 \perp}) + \mu^2 (a^{x,y}_{00} + a^{x,y}_{01} + a^{x,y}_{10} + a^{x,y}_{11}) \right) \\
    &= \frac{1}{2^{2n}} \sum_{x,y} \left( a_{\perp\perp}^{xy} + \mu \chi^{xy} + \mu^2 \tau^{xy} \right) = \frac{1}{2^{2n}} \sum_{x,y} (a_{\perp\perp}^{xy}) + \mu \chi + \mu^2 \tau,
\end{align}
where we defined $\chi^{xy} \coloneqq (a^{x,y}_{\perp 0} + a^{x,y}_{0 \perp} + a^{x,y}_{\perp 1} + a^{x,y}_{1 \perp})$, $\tau^{xy} \coloneqq (a^{x,y}_{00} + a^{x,y}_{01} + a^{x,y}_{10} + a^{x,y}_{11})$ and replaced $\frac{1}{2^{2n}} \sum_{x,y} \chi^{xy} = \chi$, and for $\tau$ equivalently. By assumption we have that the loss rate is:
\begin{align} \label{eq:lossrate}
    \mathbb{P}[\perp(\rho^{x^*y^*})] = \frac{1}{2^{2n}} \sum_{x,y} a^{xy}_{\perp\perp} = 1 - \eta_P + \delta,
\end{align}
and the complement in which at least one of the attackers answers is:
\begin{align} \label{eq:chitaurelation}
    \frac{1}{2^{2n}} \sum_{x,y} (\chi^{xy} + \tau^{xy} ) = \chi + \tau.
\end{align}
When we add equations \eqref{eq:lossrate} and \eqref{eq:chitaurelation} we get the sum over all possible events so their sum is equal to $1$. This allows us to rewrite $\chi$:
\begin{align}
    1 - \eta_P + \delta + \chi + \tau &= 1 \\
    \chi = \eta_P - \delta - \tau. \label{eq:taudefenition}
\end{align}
Now we can pick $\mu$ to enforce the proper loss rate of $1 - \eta_P$, combining the above equation with \eqref{eq:probabilityoflosspitilde} we get:
\begin{align}
    1-\eta_P &= \frac{1}{2^{2n}} \sum_{x,y} \tr{ \left(\tilde{\Pi}^{A,xy}_\perp \otimes \tilde{\Pi}^{B,xy}_\perp \right) \mathcal{E}_1^{xy}(\rho^{x^*y^*})} \\
    &= 1 - \eta_P + \delta + \mu(\eta_P - \delta - \tau) + \mu^2 \tau,
\end{align}
which reduces to
\begin{align}
    \tau \mu^2 + (\eta_P - \delta - \tau) \mu + \delta &= 0
\end{align}
Assuming $\tau \neq 0$, we can solve this quadratic equation for $\mu$ and get a positive and negative solution. Since $\mu$ is positive we can discard the negative solution to get:
\begin{align}
    \mu =  \frac{\sqrt{(\eta_P - \delta - \tau)^2 - 4 \delta \tau} - \eta_P + \delta + \tau}{2\tau}.
\end{align}
Plugging in $\delta' = - \delta$ we get: 
\begin{align} \label{eq:finalexpressionmy}
    \mu =  \frac{\sqrt{(\eta_P + \delta' - \tau)^2 + 4 \delta' \tau} - \eta_P - \delta' + \tau}{2\tau}.
\end{align}
Now Lemma \ref{lemma:taumonotone} implies that eq.\  \eqref{eq:finalexpressionmy} is monotone in $\tau$, since $\eta_P, \delta', \tau \geq 0$. From eq.\  \eqref{eq:taudefenition} we see that $\tau \leq \eta_P + \delta'$, since $\chi \geq 0$, filling this in for $\tau$ and using monotonicity we get:
\begin{align}
    \mu \leq \frac{\sqrt{4 \delta' (\eta_P + \delta')}}{2 (\eta_P + \delta')} = \frac{\sqrt{\delta'}}{\sqrt{\eta_P + \delta'}} \leq \frac{\sqrt{\delta'}}{\sqrt{\eta_P}} \leq \frac{\sqrt{\Delta_{\varepsilon}}}{\sqrt{\eta_P}}.
\end{align}
The success probability of the new POVM strategy on both conclusive answers only differs slightly:

\begin{align}
    \bigl| &\tr{\tilde{\Pi}^{xy} \mathcal{E}_1^{xy}(\rho^{x^*y^*})} - \tr{\Pi^{xy} \mathcal{E}_1^{xy}(\rho^{x^*y^*})} \bigr | \leq \| \tilde{\Pi}^{xy} - \Pi^{xy} \|_\infty \\
    &\leq \| \tilde{\Pi}^{A,xy} \otimes \tilde{\Pi}^{B,xy} - \tilde{\Pi}^{A,xy} \otimes {\Pi}^{B,xy} \|_\infty + \| \tilde{\Pi}^{A,xy} \otimes {\Pi}^{B,xy} - {\Pi}^{A,xy} \otimes {\Pi}^{B,xy} \|_\infty \\
    &\leq \| \tilde{\Pi}^{B,xy} - \Pi^{B,xy} \|_\infty + \| \tilde{\Pi}^{A,xy} - \Pi^{A,xy} \|_\infty \leq 2 \mu \leq 2 \frac{\sqrt{\Delta_\varepsilon}}{\sqrt{\eta_P}}.
\end{align}
Thus the success probability on $\rho^{x^*y^*}$ when conditioned on playing when $\delta < 0$ is bounded by:
\begin{align}
    \frac{1}{2^{2n} \eta_P} &\sum_{x,y \in \Lambda_\varepsilon} \tr{ \Pi^{xy} \mathcal{E}_1^{xy}(\rho^{x^*y^*})} \nonumber \\
    &= \frac{1}{2^{2n} \eta_P} \sum_{x,y \in \Lambda_\varepsilon} \tr{ (\Pi^{xy} - \tilde{\Pi}^{xy}) \mathcal{E}_1^{xy}(\rho^{x^*y^*})} + \frac{1}{2^{2n} \eta_P} \sum_{x,y \in \Lambda_\varepsilon} \tr{ \tilde{\Pi}^{xy} \mathcal{E}_1^{xy}(\rho^{x^*y^*})} \\
    &\leq \frac{|\Lambda_\varepsilon|}{2^{2n} \eta_P} \|\Pi^{xy} - \tilde{\Pi}^{xy}\|_1 \| \mathcal{E}_1^{xy}(\rho^{x^*y^*})\|_\infty + \mathbb{P}[\text{attack } \mathsf{P}_{\eta_P}]  \\
    &\leq \mathbb{P}[\text{attack } \mathsf{P}_{\eta_P}] + \frac{\sqrt{\Delta_\varepsilon}}{\eta_P \sqrt{\eta_P}} \frac{|\Lambda_\varepsilon|}{2^{2n}}. \label{eq:finalinequalitydeltanegative}
\end{align}
Combining the bounds on the success probability for $\delta < 0$ above and $\delta > 0$ in eq.\  \eqref{eq:finalinequalitydeltapositive} we get:
\begin{align}
    \frac{1}{2^{2n} \eta_P} &\sum_{x,y \in \Lambda_\varepsilon} \tr{ \Pi^{xy} \mathcal{E}_1^{xy}(\rho^{x^*y^*})} \leq \mathbb{P}[\text{attack } \mathsf{P}_{\eta_P}] + \max \left(\frac{\sqrt{\Delta_\varepsilon}}{ \sqrt{\eta_P}}, \frac{\Delta_\varepsilon}{( 1 - \eta_P) } \right) \frac{|\Lambda_\varepsilon|}{2^{2n} \eta_P}.
\end{align}
By assumption $\sqrt{\Delta_\varepsilon} \geq \Delta_\varepsilon$ thus, 
\begin{align}
    \frac{1}{2^{2n} \eta_P} &\sum_{x,y \in \Lambda_\varepsilon} \tr{ \Pi^{xy} \mathcal{E}_1^{xy}(\rho^{x^*y^*})} \leq \mathbb{P}[\text{attack } \mathsf{P}_{\eta_P}] + \sqrt{\Delta_\varepsilon} g(\eta_P) \frac{|\Lambda_\varepsilon|}{2^{2n} \eta_P}.    
\end{align}
\end{proof}

We can now formulate a statement about the security of a protocol with a commit round added on top of a regular protocol. This is useful because it does not give attackers the opportunity to use the option of answering `loss' very often anymore and raises the effective transmission of the protocol from $\eta_V \eta_P$ to the usually much larger $\eta_P$. The latter may be large enough to protect against lossy attacks that arise in e.g.\ $f$-BB84 QPV protocols. On the other hand, it opens up a new possible attack. Attackers can now try to apply some transformation on their state and answer `no commit' when this transformation fails. However, they still need to answer the same commitment to both verifiers. In the following theorem we show that this action cannot help them much. Because the attackers need to give the same commit-bit with very high probability, the size of $\Sigma^c_{\varepsilon}$ will be small relative to all possible inputs. Then a large number of post-measurement states will be close to a fixed state independent of $x,y$ by Lemma \ref{lemma:existence_reference_state_close_to_other_states}. We can now bound the probability of success of the protocol with commitment, because the post-measurement state can be replaced by one fixed state independent of $x,y$. Thus the attackers find themselves in the same situation as attacking the underlying protocol. Any underlying protocol that remains secure for any (constant) adversarial input state as in Definition~\ref{def:state_indep} thus has a corresponding commitment-protocol with the same security guarantee (up to a small overhead). We make this precise in the following theorem. Note that a particular protocol with the considered properties is \QPVBBf~\cite{bluhm2022single}.

\begin{theorem} \label{thm: main theorem upper bounding commit attack} 
    Let  $\mathsf P$  be a quantum position verification protocol in which the verifiers send classical and quantum information and the prover responds with classical answers. Suppose that for its version with commitment, \emph{\textsf{c-}}$\mathsf{P}$, we have $\abs{\Sigma^c_{\varepsilon}}\leq \Tilde{c} 2^{2n}$, for some $\varepsilon, \tilde{c}$ sufficiently small such that $1- (8 \sqrt{\varepsilon} + \sqrt{(1-2\tilde{c})8\sqrt{\varepsilon} + 4 \tilde{c}} \, g(\eta_P)) \geq 0$. If $\mathsf P$ is state-independent (cf.~Definition~\ref{def:state_indep}) then, on the rounds the attackers play, the following bound on the probability of attackers answering correctly to \emph{\textsf{c-}}$\mathsf{P}$ holds:    \begin{equation}\label{eq: theorem committing}
        \mathbb P[\mathrm{attack } \, \emph{\textsf{c-}}\mathsf {P}_{\eta_{V},\eta_P}] \leq \frac{1}{\eta_P} \left( (1-2\tilde{c}) \left(8 \sqrt{\varepsilon} + g(\eta_P) \sqrt{(1-2\tilde{c})8\sqrt{\varepsilon} + 4 \tilde{c}} \right)  + 2 \tilde{c} \right) + \mathbb{P}[\mathrm{attack} \, \mathsf{P}_{\eta_P}]
    \end{equation}
\end{theorem}
\begin{proof}

Both attackers need to generate a commitment bit $(c_A, c_B)$ and send it to the verifiers. The most general operation two attackers can do to generate these bits is a quantum instrument. By Lemma \ref{lemma: instrument and channel} we can split up the quantum instrument in a measurement followed by a quantum channel. Here the measurement outcome corresponds to the commitment bit the attackers generate and the quantum channel corresponds to the operation they further perform, possibly depending on their inputs $(x,y)$. We want to upper bound the attacking probability in the case both attackers commit to playing (i.e.~$c_A = c_B =1$, we denote this in the subscript of the instrument). Using the Stinespring dilation theorem we can dilate these quantum channels to unitaries over some larger quantum system and we get the following for the (renormalized) post instrument state the attackers hold if they both commit to playing:
\begin{align}
    \Tilde{\mathcal{I}}_1^{xy}(\rho) &= \frac{\mathcal{I}_1^{xy}(\rho)}{\tr{\mathcal{I}_1^{xy}(\rho)}} = \frac{\mathcal{E}_1^{xy} \left( \left(\sqrt{M^x_A} \otimes \sqrt{M^y_B} \right) \rho \left(\sqrt{M^y_B} \otimes \sqrt{M^x_A}\right) \right)}{\tr{
    \left(M^x_A \otimes M^y_B\right) \rho}} \\
    &=\mathcal{E}_1^{xy}(\rho^{xy}) \\
    &=\ptr{E}{U^{xy}(\rho^{xy}\otimes\ketbra{0}{0}_E)U^{xy\dagger}}.
\end{align}
By assumption $\abs{\Sigma^c_{\varepsilon}}\leq \Tilde{c} 2^{2n}$, so we can invoke Lemma \ref{lemma:existence_reference_state_close_to_other_states}, which says that there must be a reference pair $(x_*,y_*) \in \Sigma_\varepsilon$ such that there are at least $(1-2\Tilde{c})2^{2n}$ other pairs $(x,y) \in \Sigma_\varepsilon$ fulfilling
\begin{equation} \label{eq:closeness_post_commit_states}
	\| \rho^{x_*y_*} - \rho^{xy} \|_1 \leq 8 \sqrt{\varepsilon}.
\end{equation}

Combining both results, we get that when we apply some quantum channel depending on $(x,y)$ on both post measurement states, the outputs are still close. This follows straightforwardly from the data processing inequality for the $1$-norm:
\begin{align}
	\| \mathcal{E}_1^{xy}(\rho^{xy}) - \mathcal{E}_1^{xy}(\rho^{x_*y_*}) \|_1 &\leq \norm{\rho^{xy}-\rho^{x_*y_*}}_1 \\
	&\leq 8 \sqrt{\varepsilon}.
\end{align} 
We defined $\Lambda^{(x,y)}_\varepsilon$ to be the set of all quantum states close to some reference state $\rho^{xy}$:
\begin{equation}
	\Lambda^{(x,y)}_\varepsilon = \left\{ (x',y')  \in \Sigma_\varepsilon : \| \rho^{xy} - \rho^{x'y'} \|_1 \leq 8 \sqrt{\varepsilon} \right\},
\end{equation}
and wrote $\Lambda_\varepsilon := \Lambda^{(x_*,y_*)}_\varepsilon$ for the remainder of this proof. By the previous argument we have $|\Lambda_\varepsilon| \geq (1-2\tilde{c})2^{2n}$, and $|\Lambda^c_\varepsilon| \leq 2\tilde{c} \, 2^{2n}$. Note that the last inequalities implies an upper bound for $\Delta_\varepsilon$:
\begin{align}
    \Delta_\varepsilon &= {\frac{|\Lambda_\varepsilon|}{2^{2n}}8\sqrt{\varepsilon} + \frac{2|\Lambda_\varepsilon^c|}{2^{2n}}} 
    = {\left(1-\frac{|\Lambda_\varepsilon^c|}{2^{2n}}\right)8\sqrt{\varepsilon} + \frac{2|\Lambda_\varepsilon^c|}{2^{2n}}} \\
    &= {8 \sqrt{\varepsilon} + (2-8\sqrt{\varepsilon})\frac{|\Lambda_\varepsilon^c| }{2^{2n}}} 
    \leq {8\sqrt{\varepsilon} + (2-8\sqrt{\varepsilon}) 2 \tilde{c}} \\
    &= {(1-2\tilde{c})8\sqrt{\varepsilon} + 4\tilde{c}}.
\end{align}

Now we have all the ingredients to upper bound the attacking probability of a round in which both attackers committed. Again, denote the final operation of the attackers by $\Pi^{A,(x,y)}_{AB_\text{com},\textsc{c}} \otimes \Pi^{B,(x,y)}_{A_\text{com}B,\textsc{c}} = \Pi^{xy}$. Then, conditioned on answering we have:
\begin{align}
    \mathbb P[&\mathrm{attack }\hspace{1mm} \textsf{c-}\mathsf{P}_{\eta_{V},\eta_{P}}] = \frac{1}{2^{2n} \eta_P} \sum_{(x,y)} \tr{ \Pi^{xy} \tilde{\mathcal{I}}_1^{xy}(\rho) } \\
	&= \frac{1}{2^{2n} \eta_P} \sum_{(x,y) \in \Lambda_\varepsilon} \tr{ \Pi^{xy} \mathcal{E}_1^{xy}(\rho^{xy}) } + \frac{1}{2^{2n} \eta_P} \sum_{(x,y) \in \Lambda^{c}_\varepsilon} \tr{\Pi^{xy} \mathcal{E}_1^{xy}(\rho^{xy}) } \\ 
	&\leq \frac{1}{2^{2n} \eta_P} \sum_{(x,y) \in \Lambda_\varepsilon} \tr{ \Pi^{xy} (\mathcal{E}_1^{xy}(\rho^{xy})-\mathcal{E}_1^{xy}(\rho^{x_*y_*}) + \mathcal{E}_1^{xy}(\rho^{x_*y_*}) ) } + \frac{|\Lambda^c_\varepsilon|}{2^{2n} \eta_P} \\
	&= \frac{1}{2^{2n} \eta_P} \sum_{(x,y) \in \Lambda_\varepsilon} \tr{ \Pi^{xy} (\mathcal{E}_1^{xy}(\rho^{xy})-\mathcal{E}_1^{xy}(\rho^{x_*y_*}))} + \frac{1}{2^{2n} \eta_P} \sum_{(x,y) \in \Lambda_\varepsilon} \tr{ \Pi^{xy} \mathcal{E}_1^{xy}(\rho^{x_*y_*}) ) } + \frac{|\Lambda^c_\varepsilon|}{2^{2n} \eta_P} \\
	&\leq \frac{1}{2^{2n} \eta_P} \sum_{(x,y) \in \Lambda_\varepsilon} \|\Pi^{xy}\|_\infty \|\mathcal{E}_1^{xy}(\rho^{xy})-\mathcal{E}_1^{xy}(\rho^{x_*y_*})\|_1 + \frac{1}{2^{2n} \eta_P} \sum_{(x,y) \in \Lambda_\varepsilon} \tr{ \Pi^{xy} \mathcal{E}_1^{xy}(\rho^{x_*y_*}) ) } + \frac{|\Lambda^c_\varepsilon|}{2^{2n} \eta_P} \\ 
	&\leq \frac{1}{\eta_P} \left( \frac{|\Lambda_\varepsilon|}{2^{2n}} 8 \sqrt{\varepsilon} + \frac{|\Lambda^c_\varepsilon|}{2^{2n}} \right)+ \frac{1}{2^{2n} \eta_P} \sum_{(x,y) \in \Lambda_\varepsilon} \tr{ \Pi^{xy} \mathcal{E}_1^{xy}(\rho^{x_*y_*}) ) } \\
    &\leq \frac{1}{\eta_P} \left( \frac{|\Lambda_\varepsilon|}{2^{2n}} 8 \sqrt{\varepsilon} + \frac{|\Lambda^c_\varepsilon|}{2^{2n}} + \sqrt{\Delta_\varepsilon} g(\eta_P) \frac{|\Lambda_\varepsilon|}{2^{2n}} \right) + \mathbb{P}[\mathrm{attack} \, \mathsf{P}_{\eta_P}] \\
    &= \frac{1}{\eta_P} \left( \left( 1-    \frac{|\Lambda^c_\varepsilon|}{2^{2n}} \right) (8 \sqrt{\varepsilon} + \sqrt{\Delta_\varepsilon} g(\eta_P)) + \frac{|\Lambda^c_\varepsilon|}{2^{2n}}  \right) + \mathbb{P}[\mathrm{attack} \, \mathsf{P}_{\eta_P}] \\
    &\leq \frac{1}{\eta_P} \left( (8 \sqrt{\varepsilon} + \sqrt{\Delta_\varepsilon} g(\eta_P)) + 2 \tilde{c}(1- (8 \sqrt{\varepsilon} + \sqrt{\Delta_\varepsilon} g(\eta_P)) \right) + \mathbb{P}[\mathrm{attack} \, \mathsf{P}_{\eta_P}] \\
    &= \frac{1}{\eta_P} \left( (1-2\tilde{c}) (8 \sqrt{\varepsilon} + \sqrt{\Delta_\varepsilon} g(\eta_P)) + 2 \tilde{c} \right) + \mathbb{P}[\mathrm{attack} \, \mathsf{P}_{\eta_P}] \\
    &\leq \frac{1}{\eta_P} \left( (1-2\tilde{c}) \left(8 \sqrt{\varepsilon} +  g(\eta_P)\sqrt{(1-2\tilde{c})8\sqrt{\varepsilon} + 4\tilde{c}} \right) + 2 \tilde{c} \right) + \mathbb{P}[\mathrm{attack} \, \mathsf{P}_{\eta_P}] 
\end{align}
where we used the triangle inequality, H\"older's inequality for Schatten norms \cite{watrous_theory_2018}, Lemma \ref{lemma:upperboundforrhoxstarystar}, $(1- (8 \sqrt{\varepsilon} + \sqrt{\Delta_\varepsilon} g(\eta_P))) \geq 0$, which is true if $\varepsilon$ is picked sufficiently small, and lastly the before mentioned upper bound on $\Delta_\varepsilon$. 
\end{proof}

The idea is now to estimate $\varepsilon$ and $\Tilde{c}$ to show that over an increasing number of rounds, $\mathbb P[\mathrm{attack } \, \textsf{c-}\mathsf {P}_{\eta_{V},\eta_P}]$ becomes increasingly closer to $\mathbb P[\mathrm{attack } \, \mathsf P_{\eta_{P}}]$. This should follow from getting better and better estimates of $\varepsilon$ when verifiers keep on seeing only equal commitments.

The sequentially repeated protocol, denoted by $\textsf{c-}\mathsf{P}^{\mathrm{seq}}_{\eta_{V},\eta_P}$, works as follows:
\begin{enumerate}
    \item The verifiers collect a certain number of rounds $r$ of $\textsf{c-}\mathsf{P}_{\eta_{V},\eta_P}$ that come back with commitments $(c_A,c_B) \neq (0,0)$, as detailed below for the non-adaptive and adaptive case. Rounds with $(c_A,c_B) = (0,0)$ are discarded.
    \item If in any round the verifiers see different commits, i.e.\ $(c_A,c_B) = (0,1) \text{ or } (1,0)$, or different protocol answers, they abort immediately.
    \item Otherwise, after reaching the required number of $(c_A,c_B) \neq (0,0)$ rounds, they do the security analysis as described in Section~\ref{sec:SeqRep} and accept or reject, depending on the score $\Gamma_r$ of the sample.
\end{enumerate}

\subsection{Parameter estimation}

\subsubsection{Non-adaptive strategies}\label{subsect:non_adaptative}

The above theorem gives us a way to bound the probability of success in any lossy setting, which makes protocols with a commitment round ideal candidates for practical implementation of QPV. The role of $\varepsilon$ and $\Tilde{c}$ are important here. Theoretically, if we set $\varepsilon$ to $0$, i.e.~we never allow attackers to answer different commits, we see that the attackers cannot exploit the commitment step at all! Thus the protocol with commitment is fully loss tolerant against transmission loss $1-\eta_V$ in this case.

However, as we have shown before we cannot set $\varepsilon$ to be 0, since a small $\varepsilon$ might help the attackers, while still not being detected with high probability. On the other hand, if we play a certain number of rounds in which we see a sufficient amount of committing rounds, but never see different commit bits being sent, we can be quite certain that the probability of one party not committing given that the other party commits is small. Thus, we want to estimate the conditional probabilities:
\begin{align} \label{eq:conditional_probabilites}
    \mathbb{P}[ c_A = 0 | c_B = 1] &= \frac{1}{2^{2n}} \sum_{x,y} \mathbb{P}[ c_A = 0 | c_B = 1, x_A,y_B], \\
    \mathbb{P}[ c_B = 0 | c_A = 1] &= \frac{1}{2^{2n}} \sum_{x,y} \mathbb{P}[ c_B = 0 | c_A = 1, x_A,y_B]  \label{eq:conditional_probabilites2}.
\end{align}

Intuitively, if we see a large number of rounds in which both parties commit but we never see different commits, these probabilities should be small. Suppose we want to upper bound the maximum conditional probability of the two in eq.~\eqref{eq:conditional_probabilites} by some value $\alpha > 0$. Then we can do the following. We keep playing until we get $\frac{r}{\alpha}$ number of rounds in which both parties commit, where $r$ is some fixed constant. This takes an expected number $\frac{r}{\alpha \, p_{\mathrm{commit}}}$ of rounds, where $p_\mathrm{commit}$ is the probability that the honest prover will commit. 

Suppose the attackers' strategy is non-adaptive. Then, if we detect different commit bits in one of these rounds we immediately abort, because an honest prover would never send these. If the probability of answering different commit bits was larger than $\alpha$, the probability to answer equal commit bits (and not get detected) every round in which they commit would be smaller than $(1-\alpha)^{\frac{r}{\alpha}}$.

We will now lower bound the probability to detect attackers due to differing commits. Suppose the maximum of the two probabilities eqs.~\eqref{eq:conditional_probabilites}, \eqref{eq:conditional_probabilites2} is at least $\alpha$ and denote the events $C^i_\mathrm{diff} = \{ (c_A^i, c_B^i) = (0,1) \text{ or } (1,0) \}$, $C^i_\mathrm{eq} = \{ (c_A^i, c_B^i) = (0,0) \text{ or } (1,1) \}$, $C^i_\mathrm{(1,1)} = \{ (c_A^i, c_B^i) = (1,1) \}$ and $C^i_{\neq 0} = \{ (c_A^i, c_B^i) \neq (0,0) \}$. Then for $i, j \in \{ 1, \dots, r/\alpha \}$ attackers are detected due to differing commits with probability
\begin{align}
    \mathbb{P}[\text{detect attackers}\,|\,\text{commits} \neq (0,0)] &= \mathbb{P}[\exists j \text{ with } (c_A^j, c_B^j) = (0,1) \text{ or } (1,0) \, | \, \forall i \quad (c_A^i, c_B^i) \neq (0,0)] \\
    &= \mathbb{P}[\exists j \text{ with } C^j_\mathrm{diff} \,|\, \forall i \,\, C^i_{\neq 0}].
\end{align}
Using the complementary probability and the fact that attackers act non-adaptively, we can write
\begin{align}
    \mathbb{P}[\text{detect attackers}\,|\,\text{commits} \neq (0,0)] &= 1 - \mathbb{P}[ \forall i \,\, C^i_\mathrm{eq} \, | \, \forall i \,\, C^i_{\neq 0}] \\
    &= 1 - \prod_{i=1}^{r/\alpha} \mathbb{P}[C^i_\mathrm{(1,1)} \, | \, C^i_{\neq 0}] = 1 - \prod_{i=1}^{r/\alpha} \left( 1 - \mathbb{P}[C^i_\mathrm{diff} \, | \, C^i_{\neq 0}] \right) \label{equ:P_detect_prob1} \\
    &\geq 1 - \prod_{i=1}^{r/\alpha} \left( 1 - \max\{ \mathbb{P}[c_B^i = 0 \, | \, c_A^i = 1], \mathbb{P}[c_A^i = 0 \, | \, c_B^i = 1] \} \right) \label{equ:P_detect_prob2} \\
    &\geq 1 - \prod_{i=1}^{r/\alpha} \left( 1 - \alpha \right) = 1 - (1-\alpha)^{r/\alpha} \\
    &\geq 1 - e^{-\alpha r/\alpha} = 1 - e^{-r}.
\end{align}
In the second equality, we use that $C^i_\mathrm{eq} \cap \{C^j_{\neq 0} \forall j\} = C^i_{(1,1)} = C^i_{(1,1)} \cap C^i_{\neq 0}$ and that the attacks are non-adaptive.
The first inequality follows from the following argument. Notice that the event $\{ (c_A^i, c_B^i) \neq (0,0) \}$ contains $ \{ c_A^i = 1 \text{ or } c_B^i = 1 \}$. Consider the case of $c_A^i = 1$. Then we can write
\begin{align}
    \mathbb{P}[C^i_\mathrm{diff} \, | \, c_A^i = 1] &= \frac{\mathbb{P}[(c_A^i, c_B^i) = (1,0)]}{\mathbb{P}[(c_A^i, c_B^i) = (1,0)] + \mathbb{P}[(c_A^i, c_B^i) = (1,1)]}, \label{equ:probs_proof_nonadaptive}\\
    \mathbb{P}[C^i_\mathrm{diff} \, | \, C^i_{\neq 0}] &= \frac{\mathbb{P}[(c_A^i, c_B^i) = (1,0)] + \mathbb{P}[(c_A^i, c_B^i) = (0,1)]}{1-\mathbb{P}[(c_A^i, c_B^i) = (0,0)]} \label{equ:probs_proof_nonadaptive2}.  
\end{align}
Writing $a = \mathbb{P}[(c_A^i, c_B^i) = (0,0)], b = \mathbb{P}[(c_A^i, c_B^i) = (0,1)], c = \mathbb{P}[(c_A^i, c_B^i) = (1,0)]$ and $d = \mathbb{P}[(c_A^i, c_B^i) = (1,1)]$ on can directly verify that $\frac{c}{c+d} \leq \frac{c+b}{1-a}$ given that $a+b+c+d=1$. Thus
\begin{align}
     \mathbb{P}[C^i_\mathrm{diff} \, | \, C^i_{\neq 0}] \geq \mathbb{P}[C^i_\mathrm{diff} \, | \, c_A^i = 1] = \mathbb{P}[c_B^i = 0 \, | \, c_A^i = 1].
\end{align}
The case $c_B^i = 1$ works the same way. Hence
\begin{align}
    \mathbb{P}[C^i_\mathrm{diff} \, | \, C^i_{\neq 0}] \geq \max\{ \mathbb{P}[c_B^i = 0 \, | \, c_A^i = 1], \mathbb{P}[c_A^i = 0 \, | \, c_B^i = 1] \}.
\end{align}
We see that if the probability to commit differently was higher than $\alpha$ we would detect attackers in the $\frac{r}{\alpha}$ committed rounds with probability exponentially close to $1$ in $r$. When we pick $r=20$, we have that $\mathbb{P}[\text{detect attackers}\,|\,\text{commits} \neq (0,0)] \geq 1 - 10^{-9}$.  And, if we do not see any different commit bits in $\frac{r}{\alpha}$ rounds we can say with very high probability that the probabilities in eq.~\eqref{eq:conditional_probabilites}, \eqref{eq:conditional_probabilites2} are upper bounded by $\alpha$. The more rounds we run, the smaller we can make $\alpha$ (with high probability), thus controlling the role of $\varepsilon$ in Theorem \ref{thm: main theorem upper bounding commit attack}.

For the theorem to be of any use, we also need to control the dependence on $\Tilde{c}$ (which comes from $\abs{\Sigma^c_{\alpha}}\leq \Tilde{c} 2^{2n}$). Intuitively, if the set $\Sigma_\alpha^c$ is large, we know that a big part of this set must be close to $\alpha$ in order for the average over all probabilities to still be $\alpha$. Then, if we would look at, e.g.\ $\Sigma_{2\alpha}^c$, we expect the set to be much smaller. We can make this intuition precise. Suppose we play $k'\frac{20}{\alpha}$ number of rounds for some value $\alpha$ that we fix beforehand. Then by the previous argument we can assume with high probability that $\max \{\mathbb{P}[ c_A = 0 | c_B = 1], \ \mathbb{P}[ c_B = 0 | c_A = 1] \} \leq \frac{\alpha}{k'}$. Then consider the set $\Sigma_\alpha^c$. In the worst case, all the values in this set are very close to $\alpha$ and, in order for the average to be $\frac{\alpha}{k'}$, we get that the maximal size is $|\Sigma_\alpha^c| \leq \frac{2}{k'} 2^{2n}$. Indeed, from the condition that $\max \{\mathbb{P}[ c_A = 0 | c_B = 1], \ \mathbb{P}[ c_B = 0 | c_A = 1] \} \leq \frac{\alpha}{k'}$ it follows that in the worst case both probabilities are equal to $\alpha/k'$ and have non-zero values on disjoint pairs of $(x,y)$. More formally, from the definition of $\Sigma_\alpha$ we know that either $\mathbb{P}[ c_A = 0 | c_B = 1, x,y] \geq \alpha$ for at least $|\Sigma_\alpha^c|/2$ pairs $(x,y)$ in $\Sigma_\alpha^c$ or $\mathbb{P}[ c_B = 0 | c_A = 1, x,y] \geq \alpha$ for at least $|\Sigma_\alpha^c|/2$ pairs $(x,y)$ in $\Sigma_\alpha^c$. Let us assume without loss of generality that we are in the former case.  We estimate 
\begin{align*}
    \frac{\alpha}{k'} & \geq \frac{1}{2^{2n}} \sum_{x,y} \mathbb{P}[ c_A = 0 | c_B = 1, x_A,y_B] \\
     & \geq \frac{1}{2^{2n}} \sum_{(x,y) \in \Sigma_\alpha^c} \mathbb{P}[ c_A = 0 | c_B = 1, x_A,y_B] \\
      & \geq \frac{1}{2^{2n}} \frac{|\Sigma_\alpha^c|}{2} \alpha
\end{align*}

 Thus, we can set $\Tilde{c} = \frac{2}{k'}$. For simplicity of the final statement from Theorem \ref{thm: main theorem upper bounding commit attack}, note that we have the freedom to pick $\alpha$ as we like. Pick $\alpha$ such that $8 \sqrt{\alpha} = 2 \tilde{c}$. Plugging in these values in Theorem \ref{thm: main theorem upper bounding commit attack} we get:
\begin{align}
    \mathbb P[\mathrm{attack } \, \emph{\textsf{c-}}\mathsf {P}_{\eta_{V},\eta_P}] &\leq \frac{1}{\eta_P} \left( (1-2\tilde{c}) \left(8 \sqrt{\varepsilon} + g(\eta_P) \sqrt{(1-2\tilde{c})8\sqrt{\varepsilon} + 4 \tilde{c}} \right)  + 2 \tilde{c} \right) + \mathbb{P}[\mathrm{attack} \, \mathsf{P}_{\eta_P}] \\
    &\leq \frac{1}{\eta_P} \left( (1-2\tilde{c}) \left(8 \sqrt{\alpha} + g(\eta_P) \sqrt{(1-2\tilde{c})8\sqrt{\alpha} + 4  \tilde{c}} \right)  + 2\tilde{c} \right) + \mathbb{P}[\mathrm{attack} \, \mathsf{P}_{\eta_P}] \\ 
    &= \frac{1}{\eta_P} \left((1-2\tilde{c})\left(2\tilde{c} + g(\eta_P) \sqrt{( 1 - 2\tilde{c})2\tilde{c} + 4\tilde{c}}\right) + 2\tilde{c} \right)+ \mathbb{P}[\mathrm{attack} \, \mathsf{P}_{\eta_P}] \\
    &\leq \frac{1}{\eta_P} \left((1-2\tilde{c})\left(2\tilde{c} + g(\eta_P) \sqrt{6\tilde{c}}\right) + 2\tilde{c} \right)+ \mathbb{P}[\mathrm{attack} \, \mathsf{P}_{\eta_P}] \\
    &\leq \frac{1}{\eta_P} \left( 3g(\eta_P)\sqrt{\tilde{c}} + 4\tilde{c} - 6g(\eta_P)\tilde{c}\sqrt{\tilde{c}} - 4\tilde{c}^2 \right)+ \mathbb{P}[\mathrm{attack} \, \mathsf{P}_{\eta_P}].
\end{align}
Then pick $\tilde{c}= \frac{1}{k^2(4+3g(\eta_P))^2}$, i.e., $k' = 2 k^2(4+3g(\eta_P))^2$, such that $3 g(\eta_P)\sqrt{\tilde{c}} + 4\tilde{c} \leq \frac{\eta_P}{k}$, and we get
\begin{align}
    \mathbb P[\mathrm{attack } \, \emph{\textsf{c-}}\mathsf {P}_{\eta_{V},\eta_P}] \leq \mathbb{P}[\mathrm{attack} \, \mathsf{P}_{\eta_P}] + \frac{1}{k}.
\end{align}
Furthermore whenever $k \geq 1$ it implies that $(1- (8 \sqrt{\varepsilon} + \sqrt{\Delta_\varepsilon} g(\eta_P))) \geq 0$. Plugging this in Theorem \ref{thm: main theorem upper bounding commit attack} we get the following corollary for the attacking probability of a $\textit{single round}$ of the protocol:

\begin{corollary} \label{cor:final_corollary_upperbound}
    Consider a quantum position verification protocol $\mathsf{P}$, with the properties described as in Theorem~\ref{thm: main theorem upper bounding commit attack} and security under sequential repetition. Suppose we play its version with commitment \emph{\textsf{c-}}$\mathsf{P}$ until we have $640/\tilde{c}^3 = O(k^6)$ rounds in which both parties commit, picking $\tilde{c}$ as described above. If attackers use a non-adaptive strategy, then either the attackers are detected with probability bigger than $1-10^{-9}$ by means of a different commitment, or we have the following bound on the probability of attacking a single round \emph{\textsf{c-}}$\mathsf{P}$ depending only on $k$:
    \begin{align}\label{eq: theorem committing2}
        \mathbb P[\mathrm{attack } \, \emph{\textsf{c-}}\mathsf {P}_{\eta_{V},\eta_P}] 
        &\leq \mathbb P[\mathrm{attack } \, \mathsf P_{\eta_{P}}] + \frac{1}{k}.
    \end{align}
\end{corollary}

Thus, by running more rounds of the protocol we can get the probability of successfully attacking the protocol to be arbitrary close to the attacking probability in a setting with no photon loss between the verifiers and the prover. What is also important to emphasize is that there is no overhead in the procedure of getting bounds in Corollary \ref{cor:final_corollary_upperbound}, since the task of committing is separate from the rounds themselves. Each round the verifiers play gives a better bound for the probability of attack for all the previous rounds played.
 
\subsubsection{Adaptive strategies:}\label{subsect:adaptative}

The above proof assumed that attackers use the same strategy in each round. But in general they could use adaptive strategies, adjusting it each round to how they responded before. We will provide a bound for this most general scenario now. Firstly note that the statement of Theorem \ref{thm: main theorem upper bounding commit attack} can also be made for the adaptive setting. In an adaptive strategy, the measurement that determines whether the attackers will commit or not given that the other party committed can now depend on the information of the previous rounds. This may change the underlying probability of events. However the proof already considers arbitrary distributions of commitments, thus we replace $\varepsilon$ by its round-dependent version $\varepsilon_i$. The attackers may replace the quantum state by some state that depends on the information of the previous rounds, but by the state-independent property this should not change the probability of successfully attacking the protocol. Therefore we get the following corollary on the probability of attacking a specific round $i$:

\begin{corollary}\label{cor:main_adaptive}
    Consider a quantum position verification protocol $\mathsf{P}$, with the properties described as in Theorem~\ref{thm: main theorem upper bounding commit attack} and security under sequential repetition. Suppose that for its version with commitment, \emph{\textsf{c-}}$\mathsf{P}$, for a given round $i$ we have $\abs{\Sigma^c_{\varepsilon_i}}\leq \Tilde{c_i} 2^{2n}$ for some $\varepsilon_i$ and $\tilde c_i$ sufficiently small such that $(1- (8 \sqrt{\varepsilon} + \sqrt{\Delta_\varepsilon} g(\eta_P))) \geq 0$. If $\mathsf P$ is state-independent (cf.~Definition~\ref{def:state_indep}) then, if the attackers play, the following bound on the probability of attackers answering correctly on the $i$-th round of \emph{\textsf{c-}}$\mathsf{P}$ holds:
    \begin{equation}\label{eq: corollary committing}
        \mathbb P[\mathrm{attack } \, \emph{\textsf{c-}}\mathsf {P}_{\eta_{V},\eta_P}] \leq \frac{1}{\eta_P} \left( (1-2\tilde{c}_i) \left(8 \sqrt{\varepsilon_i} + g(\eta_P) \sqrt{(1-2\tilde{c}_i)8\sqrt{\varepsilon_i} + 4 \tilde{c}_i} \right)  + 2 \tilde{c}_i \right) + \mathbb{P}[\mathrm{attack} \, \mathsf{P}_{\eta_P}]
    \end{equation}
\end{corollary}

The problem is now to estimate the value of $\varepsilon_i$, which we cannot estimate for every $i$ since it can change adaptively from round to round. We will show that if we run sufficiently many rounds, and never see different commits by the attackers, that then at least a large fraction of all the $\varepsilon_i$ must have been sufficiently low. 

We can make a similar argument as in the non-adaptive case, carefully including that attackers can now condition on the past in each round. We will use the general property that
\begin{align}\label{equ:pcond_rule}
    \mathbb{P} \left[ A_1 , \dots , A_n \right] = \mathbb{P} \left[ A_1 \right] \mathbb{P} \left[ A_2 \, | \, A_1 \right] \cdot \cdot \cdot \mathbb{P} \left[ A_n \, | \, A_1, \dots, A_{n-1} \right],
\end{align}
for any events $A_1, \dots, A_n$. Consider running $r$ rounds with commitments $(c_A, c_B) \neq (0,0)$. Let $i, j \in \{ 1, \dots, r \}$. Then we can bound the probability of being detected due to differing commits as follows,
\begin{align}
    \mathbb{P}[\text{detect attackers\,|\,commits} \neq (0,0)] &= 1 - \mathbb{P}[ \forall i \,\, C^i_{\mathrm{eq}} \, | \, \forall i \,\, C^i_{\neq 0}] \\
    &= 1 - \mathbb{P}[ \forall i \,\, C^i_{(1,1)} \, | \, \forall i \,\, C^i_{\neq 0}] \label{equ:pdet_adaptive}.
\end{align}
Then eq.~\eqref{equ:pdet_adaptive} can be written as
\begin{align}
    \mathbb{P}[\text{detect attackers\,|\,commits} \neq (0,0)] = 1 - \mathbb{P}[ C^1_{(1,1)}, \dots, C^{r}_{(1,1)} \, | \, C^1_{\neq 0}, \dots, C^{r}_{\neq 0}]
\end{align}
After using eq.~\eqref{equ:pcond_rule} and noting that $C^i_{(1,1)} \cap C^i_{\neq 0} = C^i_{(1,1)}$ for any $i$, this can be rewritten as
\begin{align}
    \mathbb{P}[\text{detect attackers\,|\,commits} \neq (0,0)] &= 1 - \prod_{i=1}^{r} \mathbb{P}\left[ C^i_{(1,1)} \, \Big| \, C^1_{(1,1)}, \dots, C^{i-1}_{(1,1)}, C^i_{\neq 0}, \dots, C^{r}_{\neq 0} \right] \\ 
    &= 1 - \prod_{i=1}^{r} \left( 1 - \mathbb{P}\left[ C^i_\mathrm{diff} \, \Big| \, C^1_{(1,1)}, \dots, C^{i-1}_{(1,1)}, C^i_{\neq 0}, \dots, C^{r}_{\neq 0} \right] \right).
\end{align}
We can then consider the analogous equations to eq.~\eqref{equ:probs_proof_nonadaptive}, \eqref{equ:probs_proof_nonadaptive2}, but with all the extra events for rounds $1, \dots, i-1, i+1, \dots, r$ in the conditioning. Again, labeling these probabilities analogously with $a_i,b_i,c_i,d_i$ (cf. eq.~\eqref{equ:probs_proof_nonadaptive}, \eqref{equ:probs_proof_nonadaptive2}) we obtain the inequality $\frac{c_i}{c_i+d_i} \leq \frac{c_i+b_i}{p_i-a_i}$, where now
\begin{align}
    p_i = \mathbb{P}\left[ C^1_{(1,1)}, \dots, C^{i-1}_{(1,1)}, C^i_{\mathrm{any}}, C^{i+1}_{\neq 0}, \dots, C^{r}_{\neq 0} \right],
\end{align}
with $C^i_{\mathrm{any}} = \{ (c_A^i, c_B^i) = (0,0) \text{ or } (0,1) \text{ or } (1,0) \text{ or } (1,1)\}$. The inequality can be verified under the condition that $a_i+b_i+c_i+d_i=p_i$. This shows
\begin{align}
    \mathbb{P} [ C^i_{\mathrm{diff}} \, | \, C^1_{(1,1)}, \dots, C^{i-1}_{(1,1)}, C^i_{\neq 0}, \dots, C^{r}_{\neq 0} ] &\geq \mathbb{P} \left[ C^i_{\mathrm{diff}} \, \Big| \, C^1_{(1,1)}, \dots, C^{i-1}_{(1,1)}, \{ c_A^i = 1 \}, C^{i+1}_{\neq 0}, \dots, C^{r}_{\neq 0} \right] \\
    &=\mathbb{P} \left[ c_B^i = 0 \, \Big| \, C^1_{(1,1)}, \dots, C^{i-1}_{(1,1)}, \{ c_A^i = 1 \}, C^{i+1}_{\neq 0}, \dots, C^{r}_{\neq 0} \right].
\end{align}
The same inequality holds for the case with $A$ and $B$ swapped, as before. Thus
\begin{alignat}{3}\label{equ:max_adaptive}
    \mathbb{P}&[\text{detect attackers}\,|\,&&\text{commits} \neq (0,0)] \geq\\ 
    &1 - \prod_{i=1}^{r} \Big( 1 - \max \big\{ &&\mathbb{P} [ c_B^i = 0 \, | \, C^1_{(1,1)}, \dots, C^{i-1}_{(1,1)}, \{ c_A^i = 1 \}, C^{i+1}_{\neq 0}, \dots, C^{r}_{\neq 0} ], \notag \\
    &{}&&\mathbb{P} [ c_A^i = 0 \, | \, C^1_{(1,1)}, \dots, C^{i-1}_{(1,1)}, \{ c_B^i = 1 \}, C^{i+1}_{\neq 0}, \dots, C^{r}_{\neq 0} ] \big\} \Big) \notag.
\end{alignat}
Define $\varepsilon_i'$ to be the maximum in eq.~\eqref{equ:max_adaptive}. This quantity can be interpreted as follows. In the $i$-th round adaptive attackers have the information that in all the previous rounds they committed and that they committed equally, otherwise they would have already been caught. They also know that they have to keep playing until they have reached the desired number of non-$(0,0)$ commits.

Now there are two cases, either the probability in eq.~\eqref{equ:max_adaptive} is $\geq 1- \delta$ with some security parameter $\delta > 0$, in which case the verifiers catch an attack with high probability by means of a different commit $c_A \neq c_B$ showing up, or it is $\leq 1-\delta$. In the latter case, we still need to bound the attack success probability. Note that then $$1 - \prod_{i=1}^r (1-\varepsilon'_i) \leq 1-\delta.$$ We can rewrite the condition as $$ 0 < \delta \leq \prod_{i=1}^r (1-\varepsilon'_i) \leq e^{-\sum_{i=1}^r \varepsilon'_i}.$$ Equivalently, $\sum_{i=1}^r \varepsilon'_i \leq \ln \left( 1/\delta \right)$. Next, we will need the following lemma, saying that under such a constraint there must be enough ``good'' rounds with $\varepsilon'_i$ not too large.
\begin{lemma}\label{lem:eps_upper_adaptive}
    Let $\sum_{j=1}^r \varepsilon'_j \leq \alpha$. Then for any $0 < q < 1$ such that $qr \in \mathbb{N}$, there exists a subset $\mathcal{R} \subset \{ 1, \dots, r \}$ of size $|\mathcal{R}|=qr$ such that for all $\varepsilon'_j$ with $j\in\mathcal{R}$ we have $\varepsilon'_j \leq \frac{\alpha}{(1-q)r}$.
\end{lemma}
\begin{proof}
    Assume you cannot find $qr$ elements $\varepsilon'_j$ with $\varepsilon'_j \leq \frac{\alpha}{(1-q)r}$, given $\sum_{j=1}^r \varepsilon'_j \leq \alpha$.
    Then there would be at least $(1-q)r$ elements fulfilling $\varepsilon'_j > \frac{\alpha}{(1-q)r}$. But then $\sum_{j=1}^r \varepsilon'_j > \alpha$, a contradiction. Thus, we must be able to find $qr$ such elements and let $\mathcal{R}$ be the set of those.
\end{proof}
That is, for a fraction $q$ of the $r$ rounds we have a round-independent upper bound on the $\varepsilon'_i$ of those rounds, namely $\varepsilon'_i \leq \frac{\ln (1/\delta)}{(1-q)r}$ for $i \in \mathcal{R}$.  We are free to pick $(\delta, q, k ,r)$, and we can simplify this upper bound on $\varepsilon'_i$. We pick $q = 1 - \frac{1}{k}$, such that $\varepsilon'_i \leq \frac{k \ln (1/\delta)}{r} = \frac{\alpha}{\tilde{r}}$, where we picked $\tilde{r} = \frac{r}{\ln(1/\delta) k \alpha}$. Then, in our definition of $\Sigma_{\varepsilon_i}^c$, we put $\varepsilon_i - \alpha$. Furthermore, we can pick $\alpha$ such that $8\sqrt{\alpha} = 2 \tilde{c_i}$. Moreover, a similar argument as in the proof for Corollary~\ref{cor:final_corollary_upperbound} can be run to argue that $\Tilde{c}_i \leq 2/\tilde r$. Picking $\tilde{c_i} \leq \frac{1}{k^2(4+3g(\eta_P))^2}=:\tilde c$, i.e., $\tilde r \geq 2 k^2(4+3g(\eta_P))^2$, we get the same upper bound as in Corollary~\ref{cor:final_corollary_upperbound}, and we only have to play a factor $\mathcal O(k)$ more rounds in order to get the same upper bound on all the rounds in $\mathcal{R}$. In particular, for $k \geq 1$, we have that $(1- (8 \sqrt{\varepsilon} + \sqrt{\Delta_\varepsilon} g(\eta_P))) \geq 0$, which allows us to apply Corollary~\ref{cor:main_adaptive}. Thus, for a fraction $1-\frac{1}{k}$ of the $r = \frac{32 \ln(1/\delta)k}{\tilde{c}^3}$ rounds we have by Corollary~\ref{cor:main_adaptive} that:

\begin{align} \label{equ:adaptive_bound}
    \mathbb{P}[&\mathrm{attack }\hspace{1mm} \textsf{c-}\mathsf{P}_{\eta_{V},\eta_{P}} \text{ in round } i \in \mathcal{R}] \\
    &\leq \frac{1}{\eta_P} \left( (1-2\tilde{c}_i) \left(8 \sqrt{\varepsilon_i} + g(\eta_P) \sqrt{(1-2\tilde{c}_i)8\sqrt{\varepsilon_i} + 4 \tilde{c}_i} \right)  + 2 \tilde{c}_i \right) + \mathbb{P}[\mathrm{attack} \, \mathsf{P}_{\eta_P}] \nonumber \\
    &\leq \frac{1}{\eta_P} \left( (1-2\tilde{c}_i) \left(8 \sqrt{\frac{1}{\tilde{r}}} + g(\eta_P) \sqrt{(1-2\tilde{c}_i)8\sqrt{\frac{1}{\tilde{r}}} + 4 \tilde{c}_i} \right)  + 2 \tilde{c}_i \right) + \mathbb{P}[\mathrm{attack} \, \mathsf{P}_{\eta_P}] \nonumber \\
    &\leq \mathbb{P}[\mathrm{attack} \, \mathsf{P}_{\eta_P}] + \frac{1}{k},
\end{align}
by the same calculations as in Corollary~\ref{cor:final_corollary_upperbound}. We are free to pick $\delta$, and if we pick $\delta = e^{-20} \leq 3 \cdot 10^{-9}$ we get a strong bound on the probability to attack different attackers. We summarize our findings in the following corollary:

\begin{corollary} \label{cor:final_corollary_upperbound_adaptive} (Second inequality of Theorem~1.)
    Consider a quantum position verification protocol $\mathsf{P}$, with the properties described as in Theorem~\ref{thm: main theorem upper bounding commit attack} and security under sequential repetition. Suppose we play its version with commitment \emph{\textsf{c-}}$\mathsf{P}$ until we have $640 k/\tilde{c}^3 = O(k^7)$ rounds in which both attackers commit. Here, $\tilde c=\frac{1}{k^2(4+3g(\eta_P))^2}$ This takes an expected number of rounds $O(k^7)/p_{\mathrm{commit}}$. We call this protocol \emph{\textsf{c}}-$\mathsf{P}^{\mathrm{seq}}$. Then either the attackers are detected with probability bigger than $1-3\cdot10^{-9}$ by means of a different commitment, or there is a set $\mathcal R$ of size $1-1/k$ times the number of rounds such that
        \begin{align}\label{eq: theorem committing adaptive}
        \mathbb P[\mathrm{attack } \, \emph{\textsf{c-}}\mathsf{P}^{\mathrm{seq}}_{\eta_{V},\eta_P}\emph{\text{ in round }}i ] \leq \mathbb P[\mathrm{attack } \, \mathsf P_{\eta_{P}}] + \frac{1}{k}
    \end{align}
    for all $i \in \mathcal R$.
\end{corollary}

\section{Sequential repetition}\label{sec:SeqRep}

Throughout this section, we  consider a quantum position verification protocol $\mathsf P$ that fulfills the conditions of Theorem \ref{thm: main theorem upper bounding commit attack}. We present a binary test for the verifiers to either \emph{accept} or \emph{reject} the location. An honest prover will pass the test with high probability, showing \emph{completeness}, whereas attackers will fail it with exponentially high probability in $r$, showing \emph{soundness}. By proving both properties, we show \emph{security}. 

We will prove security for sequential repetition of \textsf{c-}$\mathsf{P}$, showing that after $r$ sequential repetitions, the probability that attackers break the protocol decays exponentially in $r$; in each repetition, we assume that constraints in the model we are analyzing apply in each round. We will analyze  the above studied security models: for $\varepsilon=\Tilde{c}=0$, non-adaptive strategies (corresponding to Section~\ref{subsect:non_adaptative}), and any adaptive strategies (corresponding to Section~\ref{subsect:adaptative}), which we will shortly denote by \textbf{S1}, \textbf{S2}, and \textbf{S3}, respectively.

For a random variable $X$, taking values on a finite set $\mathfrak{X}=\{\texttt{x}_1,...,\texttt{x}_d\}$, a probability distribution $p$ 
is specified by  $p_{\texttt{x}_i}=\Pr[X=\texttt{x}_i], \texttt{x}_i\in \mathfrak{X},$ and $p$ can be represented by a probability vector $\textbf{p}=(p_{\texttt{x}_1},...,p_{\texttt{x}_d})$. The set of all probability distributions $\textbf{p}$ over $\mathfrak X$ is $\Delta_{d-1}=\{\textbf{p}\in\mathbb R^d\mid \sum_{\texttt{x}_i\in \mathfrak{X}}p_{\texttt{x}_i}=1, p_i\geq0\}$, which is known as the probability simplex, and it is a ($d-1$)-dimensional manifold. In this section, we will make use of the Hoeffding's and Azuma's inequality; we state them below for completeness. 

\begin{lemma} 
\label{lemma:Hoeffding} \emph{(Hoeffding's inequality \cite{hoeffding1963})} Let $T_1,\ldots,T_r$ be independent bounded random variables with $T_i\in[x_a,x_b]$, for all $i\in\{1,\ldots,r\}$, with $-\infty< x_a\leq x_b<\infty$. Then, for all $\delta\geq 0$, the following holds:
\begin{equation}
    \Pr\left[\frac{1}{r}\sum_{i=1}^r(T_i-\mathbb E[T_i])\geq \delta\right]\leq e^{-\frac{2 r \delta^2 }{(x_b-x_a)^2}}.
\end{equation}
\end{lemma}

\begin{lemma}\label{Lemma:Azuma} \emph{(Azuma's inequality \cite{Azuma1967})}. Suppose $\{X_k\}_{k\geq 0}$ is a martingale or a super-martingale,  and $\abs{X_k-X_{k-1}}\leq \beta_k$ almost surely. Then, for all $N\in\mathbb N$ and all $\beta\in\mathbb R^+$,
\begin{equation}
    \Pr[X_N-X_0\geq \beta]\leq e^{-\frac{\beta^2}{2\sum_{k=1}^N\beta_k^2}}.
\end{equation}
\end{lemma}

Consider the realistic situation where $\eta_P<1$ and the verifiers are expected to receive a `photon loss' answer with probability $1-\eta_P$. Given a probability of error $p_{\mathrm{err}}$, an honest prover is expected to reproduce $\textbf{p}_{hp}=(p_\textsc{c},p_{\perp},p_{\textsc{i}})=(p_{x_1},p_{x_2},p_{x_3})$, depending on $\eta_P$ and $p_{\mathrm{err}}$, where $p_\textsc{c},p_{\perp},p_{\textsc{i}}$ denote the probability of being correct, answering `photon loss' and answering incorrectly, respectively. For example, if the error is independent of the loss, $\textbf{p}_{hp}=(\eta_P(1-p_{\mathrm{err}}),1-\eta_P,\eta_Pp_{\mathrm{err}})$.  Bounds on \textsf{c-}$\mathsf{P}_{\eta_V, \eta_{P}}$ characterize a (secure) subset $\mathcal{S}\subsetneq \Delta_2$ such the attackers do not have access to strategies reproducing probabilities in $\mathcal{S}$. Let $\mathcal{A}=\Delta_2\setminus \mathcal{S}$, which is the set that the attackers potentially have access to. In particular, it contains the set of all probabilities that the attackers have access to, which is convex, since given any two strategies, they are allowed to play their convex combination.  If the bounds on the probabilities are tight, $\mathcal{A}$ corresponds to the set of all probabilities that the attackers have access to. Security for \textsf{c-}$\mathsf{P}_{\eta_V, \eta_{P}}$ implies, in particular, that $(1,0,0)\notin \mathcal{A}$. Let $\gamma\subset\Delta_2$ be the curve that, together with the boundary of $\Delta_2$, describes $\mathcal{S}$ (cf. Figure~\ref{Fig_probability_simplex}) and assume $\gamma$ is differentiable (otherwise take an approximation of $\gamma$ contained in $\mathcal{S}$ that is differentiable). Consider the ruled surface  $F(p_\textsc{c},p_{\perp},p_{\textsc{i}})=0$ defined by the straight lines connecting every point in $\gamma$ with the origin $(0,0,0)$, see Figure~\ref{Fig_probability_simplex}. Then, we have that, with the corresponding choice of sign for $F$, 
\begin{equation}\label{eq:secure_region}
    \textbf{q}\cdot \nabla F\vert_{\textbf{q}} \leq 0 \hspace{2mm} \forall \textbf{q}\in \mathcal{A} \textrm{ and }  \textbf{p}\cdot \nabla F\vert_{\textbf{p}}> 0 \hspace{2mm} \forall \textbf{p}\in \mathcal{S}, 
\end{equation}
where $\nabla F=(\nabla F_{x_1},\nabla F_{x_2},\nabla F_{x_3})$ denotes the normalized gradient of $F$. 

Denote by $\textsc{ans}_i\in\{\textsc{c},\perp,\textsc{i}\}$ whether the answer they recorded in round $i$ was correct ($\textsc{c}$), `photon loss' ($\perp$),  or incorrect ($\textsc{i}$). Let
\begin{equation}
   T_i(\textbf{p}_i,\textsc{ans}_i):= \nabla F_{x_1}\vert_{\textbf{p}_i}\textbf{1}_{\textsc{c}}(\textsc{ans}_i)+\nabla F_{x_2}\vert_{\textbf{p}_i}\textbf{1}_{\perp}(\textsc{ans}_i)+\nabla F_{x_3}\vert_{\textbf{p}_i}\textbf{1}_{\textsc{I}}(\textsc{ans}_i)\text{ for all }i\in[r]. 
\end{equation}

For an honest prover ($hp$), the $T_i$'s are expected to be independent identically distributed, and thus, for every $i$, 
\begin{equation}
    \mathbb E[T_i^{hp}]=p_{\textsc{c}} \nabla F_x\vert_{\textbf{p}_{hp}}+p_{\perp} \nabla F_y\vert_{\textbf{p}_{hp}}+p_{\textsc{i}} \nabla F_z\vert_{\textbf{p}_{hp}}=\textbf{p}_{hp}\cdot \nabla F\vert_{\textbf{p}_{hp}}=: \mu>0,
\end{equation}
and defining  a \emph{total score}
\begin{equation}
    \Gamma_r^{hp}:=\sum_{i=1}^r T_i^{hp},
\end{equation}
we have that $\mathbb E [{\Gamma}_r^{hp}]=r\mu$.

Fix a parameter $\varepsilon_{\textrm{h}} > 0$, which determines the confidence level of the test that we introduce. That is, it will succeed with probability at least $1 - \varepsilon_{\textrm{h}}$ when interacting with an honest prover. 

\begin{definition}\label{def:test} Let $\varepsilon_{\textrm{h}}>0$. For the c-$\mathsf P$~protocol executed sequentially $r$ times, we define the acceptance test $\mathsf{T}^{r}_{\varepsilon_{\textrm{h}}}$, also referred to as the decision criterion, as follows:     the verifiers \emph{accept} the prover’s location if
\begin{equation}\label{eq:test}
    \Gamma_r \geq r (\mu-\delta),
\end{equation}
where $\delta=\sqrt{\frac{4\ln(1/\varepsilon_{\textrm{h}})}{r}}$. 
Otherwise, they \emph{reject}. 
\end{definition}

By Hoeffding's inequality~\cite{hoeffding1963}, see Lemma~\ref{lemma:Hoeffding}, an honest prover will be accepted except with probability at most $\varepsilon_{\textrm{h}}$, and therefore the test is \emph{complete}. We now show that $\mathsf{T}^{r}_{\varepsilon_{\textrm{h}}}$ is \emph{sound}, in particular, attackers will be rejected with exponentially high probability in $r$.

\begin{prop} \label{prop:seq_rep_loss} Let $\mathsf P$ be as in Theorem \ref{thm: main theorem upper bounding commit attack}, and secure against sequential repetition. Let $F(p_{c},p_{\perp},p_{{i}})=0$ be the ruled surface that separates the region of $\Delta_2$ where the attackers do not have access to, as defined above, for \emph{\textsf{c-}}$\mathsf{P}_{\eta_{V},\eta_{P}}$. Let $\eta_P$ and $p_{\mathrm{err}}$ and $\delta$ be such that $\textbf{p}_{hp}\cdot \nabla F\vert_{\textbf{p}_{hp}}=\mu>0$, where $\textbf{p}_{hp}$ is the probability vector expected from the honest party. Then, after $r$-sequential repetitions of the protocol, either the attackers are caught with different commitment with probability bigger than $1-3\cdot10^{-9}$, or the probability that the attackers are accepted in the test $\mathsf T^r_{\varepsilon_{\textrm{h}}}$ after $r$-sequential repetitions of \emph{\textsf{c-}}$\mathsf{P}$ is exponentially small:
    
 \begin{itemize}
        \item For security models \emph{\textbf{S1}} and  \emph{\textbf{S2}} (with $r$ as in Lemma~\ref{cor:final_corollary_upperbound}),    
   \begin{equation}\label{eq:sequential_rep_loss}
        \Pr\left[\Gamma_{r}^{\mathrm{att}}\geq r\mu(1-\delta)\right]\leq e^{-\frac{r}{2} \left(\mu(1-\delta)\right)^2},
    \end{equation} 
    \item for security model \emph{\textbf{S3}} (with $r$ as in Lemma~\ref{cor:final_corollary_upperbound_adaptive}),
     \begin{equation}\label{eq:sequential_rep_loss_ScM3}
        \Pr\left[\Gamma_{r}^{\mathrm{att}}\geq r\mu(1-\delta)\right]\leq e^{-\frac{r}{2} \left(\mu(1-\delta)-\frac{1}{k}\right)^2}.
    \end{equation}
    \end{itemize}
\end{prop}

\begin{proof}   
For any attackers, consider their corresponding score $\Gamma_r^{att}$. Due to \eqref{eq:secure_region}, ${\mathbb E[T_{i}^{\mathrm{att}}]=\textbf{q}\cdot \nabla F\vert_{\textbf{q}}\leq 0}$ for all $i\in[r]$ in the \textbf{S1} and \textbf{S2}.  Define $\Gamma^{\mathrm{att}}_0=0$. The process $\Gamma=(\Gamma^{\mathrm{att}}_r:r\geq 0)$ is a supermartingale relative to the filtration  $\mathcal{F}_r= \sigma( T_1^{\mathrm{att}},..., T_{r}^{\mathrm{att}})$. In fact, 
\begin{equation}
    \mathbb E[\Gamma_{r}^{\mathrm{att}}\mid \mathcal{F}_{r-1}]= \mathbb E[ T_{r}^{\mathrm{att}}\mid  \mathcal{F}_{r-1}]+\mathbb E[ \Gamma_{r-1}^{\mathrm{att}}\mid  \mathcal{F}_{r-1}]\leq  \Gamma_{r-1}^{\mathrm{att}},
\end{equation}
which is the definition of a supermartingale. The first equality is due to the linearity of the conditional expectation and the inequality is due to the fact that $\mathbb E[ T_{r}^{\mathrm{att}}\mid  \mathcal{F}_{r-1}]=\textbf{q}\cdot \nabla F\vert_{\textbf{q}} \leq 0$ for any $\textbf{q}\in\mathcal{A}$ the attackers chose at the round $r$ if it depends on the previous rounds in any way, and $ \Gamma_{r-1}^{\mathrm{att}}$ is $\mathcal{F}_{r-1}$-measurable. Since $||\nabla F||=1$, $\abs{ T^{\mathrm{att}}_i}=\max_{j\in\{1,2,3\}}\abs{\nabla F_{x_j} }\leq 1$, then, an immediate application of Azuma's inequality (Lemma \ref{Lemma:Azuma}) with $\beta_k=1$ leads to \eqref{eq:sequential_rep_loss}. 

For \textbf{S3}, let $\mathcal{R}$ be the set of indices $i\in[r]$ such that we have a bound (see Corollary \ref{cor:final_corollary_upperbound_adaptive}), which, by construction, is of size $\left(1-\frac{1}{k}\right)r$. Then, 
\begin{equation*}
    \Pr\left[\Gamma_{r}^{\mathrm{att}}\geq r\mu(1-\delta)\right]=\Pr\left[\sum_{i\in\mathcal{R}}T_i^{\mathrm{att}}\geq r\mu(1-\delta)-\sum_{i\notin\mathcal{R}}T_i^{\mathrm{att}}\right]\leq\Pr\left[\sum_{i\in\mathcal{R}}T_i^{\mathrm{att}}\geq r\big(\mu(1-\delta)-\frac{1}{k}\big)\right],
\end{equation*}
where the inequality follows from  using  $T_i^{\mathrm{att}}\leq 1$  for all $i\notin\mathcal{R}$. Then, the bound \eqref{eq:sequential_rep_loss_ScM3} follows analogously by considering the supermartingale $\Gamma_{\mathcal{R}}^{att}:=\sum_{i\in\mathcal{R}}T_i^{\mathrm{att}}$.

\end{proof}

\begin{figure}[htbp]
    \centering
\begin{tikzpicture}[scale=0.75,x=0.75pt,y=0.75pt,yscale=-1,xscale=1]

\draw    (241,312) -- (240.01,30) ;
\draw [shift={(240,27)}, rotate = 89.8] [fill={rgb, 255:red, 0; green, 0; blue, 0 }  ][line width=0.08]  [draw opacity=0] (8.93,-4.29) -- (0,0) -- (8.93,4.29) -- cycle    ;
\draw    (241,312) -- (526,312) ;
\draw [shift={(529,312)}, rotate = 180] [fill={rgb, 255:red, 0; green, 0; blue, 0 }  ][line width=0.08]  [draw opacity=0] (8.93,-4.29) -- (0,0) -- (8.93,4.29) -- cycle    ;
\draw    (241,312) -- (45.61,423.51) ;
\draw [shift={(43,425)}, rotate = 330.29] [fill={rgb, 255:red, 0; green, 0; blue, 0 }  ][line width=0.08]  [draw opacity=0] (8.93,-4.29) -- (0,0) -- (8.93,4.29) -- cycle    ;
\draw [line width=1.5]    (241,97) -- (95,395) ;
\draw [line width=1.5]    (241,97) -- (447,313) ;
\draw [line width=1.5]    (95,395) -- (447,313) ;
\draw [color={rgb, 255:red, 0; green, 0; blue, 0 }  ,draw opacity=1 ]   (131,321) .. controls (170,316) and (228,336) .. (235,364) ;
\draw [fill={rgb, 255:red, 223; green, 223; blue, 223 }  ,fill opacity=1 ]   (131,321) -- (200.06,315.35) -- (241,312) ;
\draw    (235,364) -- (241,312) ;
\draw    (168.2,322.8) -- (241,312) ;
\draw [fill={rgb, 255:red, 222; green, 212; blue, 212 }  ,fill opacity=1 ]   (184,326) -- (200.02,322.07) -- (241,312) ;
\draw    (195.2,329.8) -- (241,312) ;
\draw    (211.2,337.8) -- (241,312) ;
\draw    (219.2,343.8) -- (241,312) ;
\draw    (226.2,348.8) -- (241,312) ;
\draw    (232.2,357.8) -- (241,312) ;
\draw    (184,326) -- (209.57,263.85) ;
\draw [shift={(210.33,262)}, rotate = 112.37] [color={rgb, 255:red, 0; green, 0; blue, 0 }  ][line width=0.75]    (10.93,-4.9) .. controls (6.95,-2.3) and (3.31,-0.67) .. (0,0) .. controls (3.31,0.67) and (6.95,2.3) .. (10.93,4.9)   ;

\draw (159.9,269.54) node [anchor=north west][inner sep=0.75pt]  [font=\large,rotate=-359.26]  {$\nabla F$};
\draw (307,233.4) node [anchor=north west][inner sep=0.75pt]  [font=\LARGE]  {$\mathcal{A}$};
\draw (134,333.4) node [anchor=north west][inner sep=0.75pt]  [font=\LARGE]  {$\mathcal S$};
\draw (196.2,339.2) node [anchor=north west][inner sep=0.75pt]  [font=\Large]  {$\gamma $};
\draw (18,419.4) node [anchor=north west][inner sep=0.75pt]  [font=\large]  {$p_\textsc{c}$};
\draw (543,303.4) node [anchor=north west][inner sep=0.75pt]  [font=\large]  {$p_{\perp }$};
\draw (216,24.4) node [anchor=north west][inner sep=0.75pt]  [font=\large]  {$p_{\textsc{i}}$};
\end{tikzpicture}
\caption{2-dimensional probability simplex $\Delta_2$ with secure subset $\mathcal{S}$ defined by the curve $\gamma$ for a protocol  \textsf{c-}$\mathsf{P}_{\eta_{V},\eta_{P}}$.}\label{Fig_probability_simplex}
\end{figure}
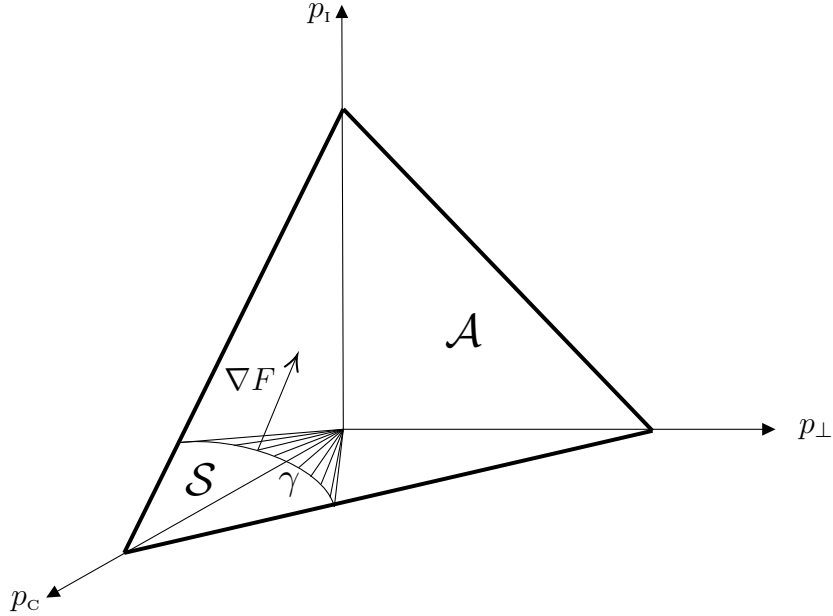

\section{Possible experimental realization of QPV with commitment}\label{sec:Practice}
For our protocol with commitment, the honest prover needs a device detecting the presence of the input quantum state\footnote{We will focus on photonic qubits.} without destroying it, i.e.~a photon presence detector, also known as quantum non-demolition (QND) measurement. We will consider two feasible ways to meet this requirement. What's important for the security of c-QPV is how much loss and error this operation introduces in the prover's setup. The main goal of c-QPV is to make the (large) transmission loss between the verifiers and the prover irrelevant for security. 

\subsubsection*{Transmission in the prover's laboratory}\label{sec:etaPsec}
The relevant transmission rate for security is the one in the prover's laboratory ($\eta_P$). It strongly depends on the actual setup used, so we will only give rough estimates of $\eta_P$. Note that
\begin{align}\label{equ:etaP1}
    \eta_P = \mathbb{P}[ \text{photon measured} \,|\, \text{presence detected}] = \frac{ \mathbb{P}[ \text{photon measured} \wedge \text{presence detected} ]}{ \mathbb{P}[ \text{presence detected} ] }.
\end{align}
The presence of a photon is concluded either due to the photon being present and detected $(\eta_V \eta_{\mathrm{det}}^{\mathrm{QND}})$ \textit{or} due to a dark count in the presence detection $(p_\mathrm{dc}^\mathrm{QND})$. Given the photon is heralded, successful measurement happens if
\begin{itemize}
    \item either the photon survived the presence detection $(\eta_\mathrm{surv})$ \textit{and} was not lost before measuring it $(\eta_\mathrm{equip})$ \textit{and} the measurement detector registered it $(\eta_\mathrm{det})$ \textit{or}
    \item (the measurement detector registered a dark count $(p_\mathrm{dc})$ when the photon did not survive the presence detection \textit{or} was lost before measurement) \textit{or} (the measurement detector registered a dark count when the presence detection also registered a dark count).
\end{itemize}
We absorb all losses after the presence detection into one term denoting the efficiency of the photon measurement $\eta_\mathrm{meas} = \eta_\mathrm{det} \eta_\mathrm{equip} \eta_\mathrm{surv}$. Using the above reasoning we can write out the probabilities in eq.~\eqref{equ:etaP1} as\footnote{For the event of a dark count it is implicit that the input photon was not detected. In our notation factors of $1-\eta_\mathrm{meas}$ or $1-\eta_V \eta_\mathrm{det}^\mathrm{QND}$ are included in the corresponding dark count variable.}
\begin{align}\label{equ:etaP2}
    \eta_P = \frac{ \left( \eta_\mathrm{meas} + p_\mathrm{dc} \right) \eta_V \eta_{\mathrm{det}}^{\mathrm{QND}} + p_\mathrm{dc} p_\mathrm{dc}^\mathrm{QND}}{\eta_V \eta_{\mathrm{det}}^{\mathrm{QND}} + p_\mathrm{dc}^\mathrm{QND}}.
\end{align}
Notice that\footnote{$p_\mathrm{dc}$ is negligible compared to the other term, so we neglect the second term in the bracket of eq.~\eqref{equ:etaP2} for eq.~\eqref{equ:etaPcases}.}
\begin{align}\label{equ:etaPcases}
    \text{if } \eta_V \ll p_\mathrm{dc}^\mathrm{QND}&: \qquad \eta_P \sim p_\mathrm{dc}.
\end{align}
If the probability that a photon enters the presence detector ($\eta_V$) is much smaller than the dark count rate $p_\mathrm{dc}^\mathrm{QND}$ then most photon presence detection events, and thus $c=1$ commitments, will be due to dark counts! Then the (e.g.\ polarization) measurement on the photon will not give a click most of the time, making $\eta_P$ very small. In the limit $\eta_V \to 0$ we obtain $\eta_P \to p_\mathrm{dc}$ as expected. Single photon detectors routinely achieve $p_\mathrm{dc} \sim 10^{-7}$ or similar per detection window \cite{hadfield2009single}. For such small $\eta_P$ the usual lossy attack of guessing the provers' measurement setting (with probability $1/m$) still works because in practice we would not be able to use a high enough number of measurement settings $m$ such that $\eta_P > 1/m$. So introducing the commitment step would not help when $\eta_V \ll p_\mathrm{dc}^\mathrm{QND}$.

Let us write $\eta_V = \gamma p_\mathrm{dc}^\mathrm{QND}$ for some constant factor $\gamma$. We define the signal-to-noise ratio of the presence detection as
\begin{align} \label{equ:SNR_QND}
    \mathtt{SNR}_{\mathrm{QND}}(\gamma) = \frac{\eta_V \eta_\mathrm{det}^\mathrm{QND}}{\eta_V \eta_\mathrm{det}^\mathrm{QND} + p_\mathrm{dc}^\mathrm{QND}} = \frac{\gamma \eta_\mathrm{det}^\mathrm{QND}}{\gamma \eta_\mathrm{det}^\mathrm{QND} + 1}.
\end{align}
We have already argued that in the case $\eta_V \ll p_\mathrm{dc}^\mathrm{QND}$ our proposal is not suitable. Let's therefore focus on the case where $\eta_V$ is at least the same order of magnitude of $p_\mathrm{dc}^\mathrm{QND}$, corresponding to $\gamma \geq 1$. Then, using that $p_\mathrm{dc}$ usually is negligibly small compared to the other quantities, we can simplify $\eta_P$ as follows,
\begin{align}
    \eta_P \sim \mathtt{SNR}_\mathrm{QND}(\gamma) \eta_\mathrm{meas}.
\end{align}
The condition that the input transmission needs to be larger than $p_\mathrm{dc}^\mathrm{QND}$ will limit the distance between the verifiers and the prover. This, however, is not a characteristic of our protocol -- it is an issue for any quantum communication protocol, as any protocol fails if most signals are noise originating from dark counts.

\subsubsection*{Distance between verifiers and prover}
The transmission law for optical fibers reads $\eta = 10^{-\alpha L/10}$ \cite{senior2009optical}, where $\alpha$ is the attenuation of the fiber in dB/km and $L$ is the fiber length in km. A standard value for current telecom-wavelength optical fibers is around $\alpha = 0.2\,$dB/km \cite{senior2009optical}, with the most sophisticated hollow-core ones achieving $\alpha = 0.09\,$dB/km \cite{petrovich2025broadband}. We can solve for $L$ and insert $\eta_V$ in terms of the presence-detection dark count rate to obtain
\begin{align}\label{equ:Lfiber}
    L = -\frac{10}{\alpha} \log_{10} \left( \gamma p_\mathrm{dc}^\mathrm{QND} \right).
\end{align}

\subsubsection*{Rate of the protocol}
There are several processes that we'd like to do at a high rate in our protocol: generating single photons, modulating their polarization state, generating EPR pairs, fast switching between measurement settings depending on $f(x,y)$, and detecting single photons. State-of-the-art equipment is able to achieve the following rates (order of magnitude) today or in the near future:

\begin{itemize}
    \item Single photon generation: MHz, in principle up to GHz \cite{migdall2020}
    \item Polarization modulation: up to GHz \cite{Li2019}
    \item EPR state generation: up to GHz, depending on pump laser power \cite{lohrmann2018, anwar2021},
    \item Switching: up to THz \cite{chai2017}
    \item Single photon detector count rate: up to GHz \cite{hadfield2009single}
\end{itemize}

Therefore, we expect the components included in our protocol can be run at least at MHz rate, and potentially at GHz rate with top equipment, albeit we acknowledge that it may be challenging to run all these processes at high rates simultaneously. The achievable rate of a setup will strongly depend on the equipment and architectures used, thus we only state current maximally achievable values here and refer to the cited articles and reviews for more details. The rate of the protocol will determine the time that is needed to reach the required number of rounds, as stated in Corollary~\ref{cor:final_corollary_upperbound_adaptive}. 

The total number of rounds $R$ that we expect to run to get $640 k \cdot (k^2(4+3g(\eta_P))^2) = O(k^7)$ rounds with commitment to play $(c=1)$ is $R = O(k^7)/p_\mathrm{commit}$. The probability $p_\mathrm{commit}$ is the probability of a successful QND measurement, which in a loss-only model would just be the transmission rate $\eta_V$. For a distance of $L = 100\,$km and attenuation of around $\alpha = 0.1\,$dB/km for hollow-core fibers (resp. $0.2\,$dB/km for standard ones), the fiber transmission law yields a rate of $\eta_V = 10\,\%$ (resp. $1\,\%$). Therefore, we can estimate $R \sim 10 \cdot  O(k^7)$ (resp. $100 \cdot  O(k^7)$). With $k$ being chosen to be around $10$ in order to keep the additive error in $\mathbb P[\mathrm{attack } \, \emph{\textsf{c-}}\mathsf{P}^{\mathrm{seq}}_{\eta_{V},\eta_P}]$ low, that means $R \sim 10^8$ (resp. $10^9$) for a single successful round of quantum position verification.

Thus, at $L=100\,$km, we would need a component rate of $100\,$MHz to $1\,$GHz to achieve a single successful round of quantum position verification per second. To suppress the probability that we accidentally accept attackers, we further need to sequentially repeat the QPV protocol $r$ times. As discussed in Section \ref{sec:SeqRep}, this suppression is exponential, so a modest amount of repetitions $r$ should suffice. Consider $r \sim 10$, which would lead to a total number of rounds $R_{\mathrm{\mathsf{c}-QPV}} \sim 10^9$ (resp. $10^{10}$) for a decisive position verification of the honest prover.

Hence, for a component rate of $10^a\,$Hz, it would take $10^{9-a}$ (resp. $10^{10-a}$) seconds to achieve the required number of rounds $R_{\mathrm{\mathsf{c}-QPV}}$ to have confidently verified the position of the honest prover against attackers. Therefore, if the components can be run at $100\,$MHz for example, 10 (resp. 100) seconds would suffice. Such a slow rate of $0.1\,$Hz (resp. $0.01\,$Hz) may be sufficient for position verification, however, as realistically speaking an object cannot move far in a few seconds.

In general, if the protocol is run at component frequency $\nu$, then the expected protocol duration $t_\mathsf{P^c}$ in seconds is therefore
\begin{align}
    t_{\mathsf{c}\text{-}\mathsf{P}}(\eta_V, \eta_P, p_\mathrm{err}, \nu, k, r) = \frac{640 k \cdot (k^2(4+3g(\eta_P))^2)}{p_\mathrm{commit}(\eta_V, p_\mathrm{err}) \cdot \nu} \cdot r.
\end{align}
Given a choice of security parameter $k$, a probability to commit $p_\mathrm{commit}(\eta_V, p_\mathrm{err})$ from the prover\footnote{Which would just be $\eta_V$ in a loss-only model and if the prover had perfect equipment.}, a loss rate $\eta_P$ in the prover lab, an achievable protocol frequency $\nu$, and a choice of sequential repetitions $r$, one can then estimate how long it takes to run the protocol with the security guarantee given in Corollary~\ref{cor:final_corollary_upperbound_adaptive}.

\subsection{True photon presence detection}
Recently, a breakthrough paper \cite{niemietz2021nondestructive} demonstrated true non-destructive detection of photonic qubits. To do so, they prepare a $^{87}$Rb atom in an optical cavity in the superposition state $\ket{+} = \left( \ket{0} + \ket{1} \right)/\sqrt{2}$, where $\ket{0}$ and $\ket{1}$ denote certain energetic states of the atom. The optical cavity is tuned such that a photon cannot enter the cavity if the atom is in state $\ket{0}$, but is allowed to enter if the state is $\ket{1}$. In that case it gets reflected from one wall before leaving the cavity again, acquiring a $\pi/2$ phase shift. This interaction adds a phase to the combined photon-atom state, i.e.\ $\ket{\psi_{\mathrm{photon}}}\ket{1} \mapsto - \ket{\psi_{\mathrm{photon}}}\ket{1}$, changing the atom state from $\ket{+}$ to $\ket{-}$. Then a rotation is applied, mapping the atomic state $\ket{+} \mapsto \ket{1}$ and $\ket{-} \mapsto \ket{0}$, after which it is measured. If the result is $0$ there was a photon interacting with the atom, if the result is $1$ there was not. This measurement thus heralds the presence of a photon in the output mode of the optical cavity, which can be sent to a polarization measurement for example. \cite{niemietz2021nondestructive} achieves the following relevant experimental parameters for their photon presence detector, which we can expect to improve in the future:

\begin{align}\label{equ:QND_params}
\begin{split}
    \text{Photon in output mode given heralding $(\eta_\mathrm{surv})$: } &\sim \text{25-55}\%, \\
    \text{Dark count rate $(p_\mathrm{dc}^\mathrm{QND})$: } &\sim 3\%, \\
    \text{Fidelity of photon in output mode: } &\sim 96\%.
\end{split}
\end{align}

\noindent Note that $\eta_\mathrm{surv}$ depends on the dark count rate and was measured using weak coherent light in \cite{niemietz2021nondestructive} rather than true single photons. We take the stated range from their Figure 3b.

Even though this technology is currently unusable for \textsf{c-}QPV due to the high dark count rate (relative to realistic $\eta_V$ over longer distances), we can expect the parameters to improve significantly in the future. A true photon presence detector such as this could therefore be a clean and viable long-term solution for \textsf{c-}QPV. 

\subsection{Simplified presence detection via partial Bell measurement}
For the near term, we consider a simplified photon presence detection based on a partial linear-optical Bell measurement. Essentially, the prover has to prepare a Bell state and teleport the input state to himself when it arrives. A conclusive\footnote{We will define which click patterns count as successful further in Figure \ref{fig:BSM}.} Bell measurement (BSM) heralds the presence of the input state, after which the prover briefly stores it until she receives the classical information $x, y$ and measures it with the appropriate setting based on $x, y$. Note that we do not require a full Bell measurement. Even just discriminating 1 out of 4 Bell states via interference at one beam splitter would be enough. The scheme in Figure~\ref{fig:BSM} \cite{weinfurter1994experimental,braunstein1995,michler1996interferometric} can distinguish 2 out of 4 Bell states, doubling the efficiency, while just using linear-optical equipment. Importantly, this scheme has first been demonstrated a long time ago \cite{michler1996interferometric} and is experimentally feasible today.

\begin{figure}[htbp]
    \centering
    \includegraphics[width=0.45\linewidth]{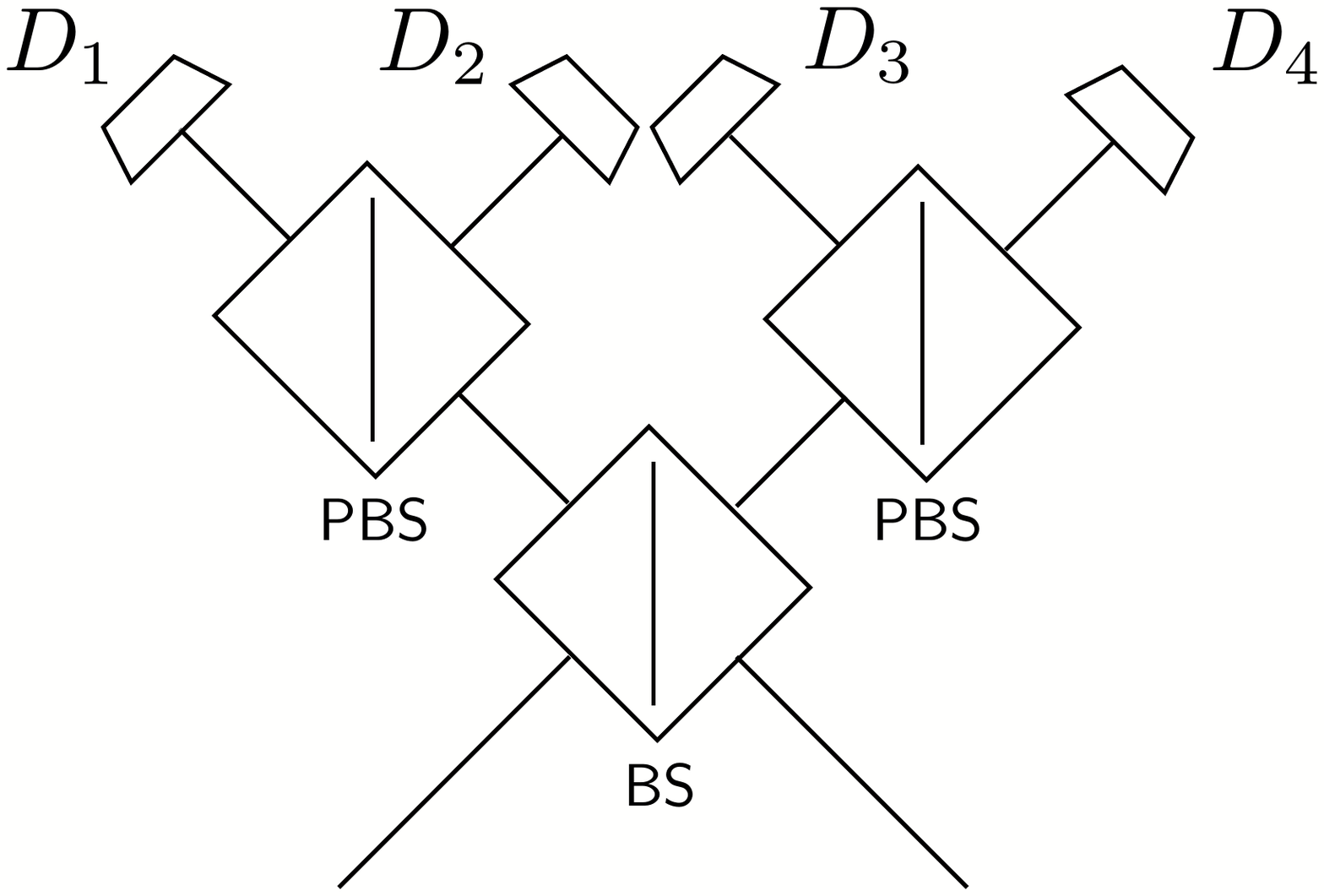}
    \caption{Schematically a partial Bell measurement can be implemented via a 50/50 beam splitter (\textsf{BS}), two polarization beam splitters (\textsf{PBS}) and four single photon detectors ($D_i$). An input state $\ket{\Psi_-}$ triggers one detector in each arm ($D_1, D_3$ or $D_2, D_4$), $\ket{\Psi_+}$ triggers two detectors in one arm ($D_1, D_2$ or $D_3, D_4$) and the states $\ket{\Phi_+}$, $\ket{\Phi_-}$ could trigger any, but just one, detector. So one can only conclusively distinguish $\ket{\Psi_-}$ and $\ket{\Psi_+}$, giving an efficiency of at most $50\%$, which is optimal for linear optics \cite{calsamiglia_maximum_2001}, but can be improved using extra auxiliary photons \cite{Ewert2014betterBSM,Bayerbach2023betterBSM} . Any click patterns other than the ones corresponding to $\ket{\Psi_\pm}$ are deemed as ``no-detection'' events.}
    \label{fig:BSM}
\end{figure}

First, note that any losses or inconclusive click patterns in the BSM itself will simply reduce the transmission $\eta_V$. This will jeopardize security only if it makes $\eta_V$ so small that dark counts take over. Moreover, it may be that the teleportation corrections do not need to be actively applied but can be classically calculated and corrected, as is the case when they just flip the measurement result predictably like in \textsf{c-}\QPVBBf~for example. So then only a partial, linear-optical BSM and (very short) storage of the other EPR qubit would be required experimentally.

If we assume that the honest prover can generate entanglement when she expects the verifiers' input to arrive, then most of the time there will be one photon (the one from the EPR pair) going into the BSM setup, and only one dark count is needed for a false positive event. The relevant photon presence detection dark count rate would then be just the one of a conventional single photon detector, i.e.~$p_\mathrm{dc}^\mathrm{QND} \sim p_\mathrm{dc}$. The presence-detection efficiency $\eta_\mathrm{det}^\mathrm{QND}$ for such a BSM would be the efficiency of detecting both photons if they are present, i.e.\ $\eta_\mathrm{det}^\mathrm{QND} = \eta_\mathrm{det}^2$. Moreover, the value of $\eta_\mathrm{meas} = \eta_\mathrm{det} \eta_\mathrm{equip} \eta_\mathrm{surv}$ depends on the equipment post-presence-detection, but is certainly upper bounded by $\eta_\mathrm{det}$. So we have an upper bound of
\begin{align}\label{equ:etaP_est}
    \eta_P \sim \mathtt{SNR}_{\mathrm{QND}}(\gamma) \eta_\mathrm{meas} \leq \frac{ \gamma \eta_\mathrm{det}^3}{\gamma 
    \eta_\mathrm{det}^2+1}.
\end{align}
Easy-to-use single photon detectors have detection efficiencies of up to $20$-$65\%$ \cite{hadfield2009single}, and the most sophisticated detectors reach up to $98\%$\footnote{Note that detection efficiencies always depend on the wavelength of the photons used.} \cite{Reddy:20}. In reality there will also be losses pre-measurement, making the true value in eq.~\eqref{equ:etaP_est} smaller than the upper bound. If these can be kept small enough, however, the true value of $\eta_P$ will be close to the upper bound in eq.~\eqref{equ:etaP_est} and secure \textsf{c-}QPV becomes possible if this value is large enough to prevent lossy attacks\footnote{Meaning higher than the basis guessing probability $1/m$ or higher than the values obtained in \cite{escolafarras2022singlequbit} for \textsf{c-}\QPVBBf, for example.}.

We summarize our findings in the following remark.
\begin{remark}\label{rem:nec_cond_cqpv}
    \textsf{c-}QPV makes a class of previously not loss-tolerant QPV protocols, with \QPVBBf~as a prime example, loss-tolerant even in practice as long as both the signal-to-noise ratio of the photon presence detection $\mathtt{SNR}_{\mathrm{QND}}$ and the efficiency of the prover measurement $\eta_\mathrm{meas}$ are sufficiently high such that $\eta_P$ is high enough to prevent lossy attacks\footnote{For example as studied in \cite{escolafarras2022singlequbit} for \QPVBBf, which carries over to our \textsf{c-}\QPVBBf.}. The signal-to-noise ratio $\mathtt{SNR}_{\mathrm{QND}}$ depends on the transmission $\eta_V$ between the verifiers and the prover, the dark count rate $p_\mathrm{dc}^\mathrm{QND}$, and the detection efficiency $\eta_\mathrm{det}^\mathrm{QND}$. This ultimately limits the maximal distance between the verifiers and the prover\footnote{To much larger distances than previously possible for QPV, however.}. The experimental requirements of our proposal in the prover laboratory are:
    \begin{itemize}
        \item The prover needs to be able to generate an EPR pair when she expects the input qubit to arrive, but this process is allowed to only work probabilistically
        \item Photon presence detection, e.g.~via a partial BSM (like the scheme in Figure~\ref{fig:BSM})
        \item A short delay loop so the prover can store the teleported qubit until the classical information $x,y$ arrives. This time delay shall be made as short as possible.
        \item The prover needs to be able to do the measurement depending on $x,y$ and should be able to quickly switch between different measurements based on the value of $f(x,y)$.
    \end{itemize}
    The verifiers need to be able to generate and modulate single photon states (e.g.\ polarization) with high frequency.
    
    All requirements are practically feasible, or within reach, with state-of-the-art equipment.
\end{remark}

\section{Discussion}
The three major roadblocks for practically implementable and secure QPV are: entangled attackers, slow honest quantum communication and signal loss. On top of that, the honest protocol must be experimentally feasible. So far, no QPV protocol was able to deal with all of those issues. Our work presents the first such protocol: \textsf{c-}\QPVBBf. This opens up a feasible route to the first experimental demonstration of a QPV protocol that remains secure in a practical setting over long distances. We propose two options to do the required non-demolition photon presence detection: a clean and viable long-term solution \cite{niemietz2021nondestructive}, assuming the non-destructive detector parameters will improve in the future, and a simpler near-term solution via a partial Bell state measurement \cite{michler1996interferometric} that can be implemented with just a few linear-optical components and conventional click/no-click single photon detectors. Given a sufficiently low dark-count rate in the photon presence detection and sufficiently low loss in the prover's laboratory, secure QPV can be achieved in principle. \textsf{c-}\QPVBBf~has two further major advantages: the quantum resources required for an attack scale in the classical input size (which can easily be made very large) and in case the prover uses the partial Bell measurement for photon presence detection, he does not need to actively apply any teleportation corrections, but can passively calculate and correct them instead, as they predictably flip the measurement outcome. By analyzing the rounds in which both attackers commit we find that when we run enough rounds attacking the committing version of the protocols becomes as hard as the underlying protocol. It would be interesting if we can use the fact that it is also difficult for attackers to always answer equally on `no commit' rounds in the analysis to get better bounds on the number of rounds we have to run. We argue that all the experimental requirements are in principle feasible and that in principle our protocol can be run at high rates. These properties taken together make \textsf{c-}\QPVBBf~the first QPV protocol that can successfully deal with all the major practical issues of QPV. 

Our result is not limited to \QPVBBf~per se, but can be applied to any QPV protocol that shares the same structure as \QPVBBf~and remains secure if the input state is replaced by any adversarial input state not depending on the classical input information $x,y$. It would be interesting to investigate whether our modification, introducing a prover commitment to play, can find application for other types of QPV protocols, or whether it can make other security models, like the random oracle model \cite{unruh_quantum_2014}, loss-tolerant.

\subsubsection*{Acknowledgments}
We thank Adrian Kent for an interesting initial discussion, pointing out photon presence detection to us. We further thank Wolfgang Löffler and Kirsten Kanneworff for helpful discussions. RA and HB were supported by the Dutch Research Council (NWO/OCW), as part of the Quantum Software Consortium programme (project number 024.003.037). AB is supported by the French National Research Agency in the framework of the ``France 2030” program (ANR-11-LABX-0025-01) for the LabEx PERSYVAL. MC acknowledges financial support from the Novo Nordisk Foundation (Grant No. NNF20OC0059939 ‘Quantum for Life’), the European Research Council (ERC Grant Agreement No. 81876) and VILLUM FONDEN via the QMATH Centre of Excellence (Grant No.10059). This research was conducted when RA and PVL were affiliated with CWI Amsterdam and QuSoft. PVL and HB were supported by the Dutch Research Council (NWO/OCW), as part of the NWO Gravitation Programme Networks (project number 024.002.003). PVL is also supported by France 2030 under the French National Research Agency award number ANR-22-PETQ-0007. LEF and FS were supported by the Dutch Ministry of Economic Affairs and Climate Policy (EZK), as part of the Quantum Delta NL programme. Part of this work was completed while MC was Turing Chair for Quantum Software, associated to the QuSoft research center in Amsterdam, acknowledging financial support by the Dutch National Growth Fund (NGF), as part of the Quantum Delta NL visitor programme.

\bibliographystyle{alphaurl}
\bibliography{biblio.bib}

\end{document}